%% file: arxiv_main.tex
\pdfoutput=1

\documentclass[11pt]{article}

\usepackage[margin=1in]{geometry}

\usepackage[utf8]{inputenc}
\usepackage[T1]{fontenc}

\usepackage{amsmath,amsfonts,amssymb,amsthm,bm}
\newtheorem{theorem}{Theorem}[section]
\newtheorem{corollary}{Corollary}[theorem]
\newtheorem{lemma}[theorem]{Lemma}

\newtheorem{remark}[theorem]{Remark}
\newtheorem{assumption}[theorem]{Assumption}

\usepackage{graphicx}
\usepackage{caption}
\usepackage{booktabs}
\usepackage{multirow}
\usepackage{makecell}
\usepackage{diagbox}
\usepackage{wrapfig}
\usepackage[table]{xcolor}
\usepackage{algorithm}
\usepackage{algorithmic}

\usepackage{nicefrac}
\usepackage{microtype}
\usepackage{url}
\usepackage{xcolor}
\usepackage[hidelinks]{hyperref}

\input{macros}

\title{Dynamic MRI Reconstruction Via\\ Dual Deep Priors and Low-Rank Plus Sparse Modeling}

\author{
Yongliang Sun$^{1,*}$ \quad
Siddhant Gautam$^{1,*}$ \quad
Chaoyan Huang$^{1,2}$ \quad
Nicole Seiberlich$^{3,4}$ \\
Ismail Alkhouri$^{5,6,\dagger}$ \quad
Saiprasad Ravishankar$^{1,7,\dagger}$ \\
\\
$^{1}$Department of Computational Mathematics, Science, \& Engineering, Michigan State University \\
$^{2}$Department of Electrical \& Computer Engineering, University of Michigan \\
$^{3}$Department of Radiology, University of Michigan \\
$^{4}$Department of Biomedical Engineering, University of Michigan \\
$^{5}$X Computational Physics Division, Los Alamos National Laboratory \\
$^{6}$Michigan Institute for Computational Discovery \& Engineering, University of Michigan\\
$^{7}$Department of Biomedical Engineering, Michigan State University \\
$^{*}$Equal first authorship~~ $^{\dagger}$Equal last authorship \\
\\
\texttt{\{sunyongl,gautamsi,huang345,ravisha3\}@msu.edu} \\
\texttt{ismailal@umich.edu} \quad
\texttt{nse@med.umich.edu} 
}

\date{}

\begin{document}

\maketitle

\begin{abstract}
Dynamic MRI reconstruction from undersampled measurements is a challenging inverse problem that requires preserving both spatial reconstruction quality and temporal consistency across the frames of the cine series. While recent learning-based approaches achieve strong performance, they heavily rely on large training, mostly fully sampled, datasets, and may otherwise generalize poorly. In contrast, training-data-free methods such as deep image prior (DIP) adapt directly to individual scans but often fail to fully exploit temporal structure and are prone to overfitting. They are particularly attractive for dynamic MRI due to the limited large, public, high-quality datasets.
In this work, we propose a structured DIP framework for dynamic MRI reconstruction that explicitly models spatiotemporal correlations through a low-rank plus sparse (L+S) decomposition. Instead of directly reconstructing the cine image series, we parameterize the low-rank background and sparse dynamic components using two DIP untrained convolutional neural networks, jointly optimized using accelerated extrapolated ADMM (eADMM). This formulation combines the implicit regularization of DIP with the interpretability of classical L+S regularization. We provide a convergence analysis for the proposed eADMM algorithm in the presence of DIP-based nonconvex parameterizations. In particular, we establish a sufficient descent property and show that every cluster point of the generated sequence is a critical point of the associated Lyapunov function. Across various acceleration factors, our numerical results demonstrate that the proposed method consistently outperforms classical reconstruction and existing supervised and unsupervised MRI reconstruction techniques.
\end{abstract}

\input{sections/introduction}
\input{sections/Preliminaries}
\input{sections/methods}

\input{sections/experiments}

\input{sections/discussions}

\section*{Acknowledgments}
This work was supported in part by the
National Science Foundation (NSF) grant CCF-2212065, NSF CAREER Award CCF-2442240, and the National Institutes of Health (NIH) Grant R21-EB030762. 

\bibliographystyle{abbrv}
\bibliography{references}

\clearpage
\appendix
\input{sections/appendix}

\end{document}

%% file: macros.tex
\usepackage{amsmath, amssymb, amsfonts, bm, mathtools}
\usepackage{xspace}



%% file: sections/introduction.tex
\section{Introduction}
Magnetic resonance imaging (MRI) is widely used in cardiovascular imaging because it provides excellent soft-tissue contrast without ionizing radiation~\cite{liang2000principles, bernstein2004handbook}. 
However, MRI acquisition is inherently time-consuming because data are collected sequentially in $k$-space, motivating accelerated imaging techniques~\cite{tsao2010ultrafast, lustig2008compressed}. 
This limitation is especially important in dynamic MRI, where multiple image frames must be recovered over time, making acceleration essential for reducing scan time, motion sensitivity, and patient burden~\cite{tsao2003k, otazo2015low}.

Classical accelerated MRI reconstructs images from undersampled data using physical encoding models and structured priors, with compressed sensing exploiting sparsity in incomplete $k$-space measurements~\cite{lustig2007sparse, lustig2008compressed}. 
For dynamic MRI, temporal redundancy has been exploited through $k$-$t$ FOCUSS~\cite{jung2009k}, blind compressed sensing~\cite{lingala2013blind}, 
$k$-$t$ SLR~\cite{lingala2011accelerated}, and low-rank plus sparse ($L+S$) modeling~\cite{otazo2015low}. 
Although interpretable, these hand-crafted priors may be restrictive for complex spatiotemporal cardiac dynamics.

Deep learning methods address this limitation by learning data-driven spatiotemporal priors from fully sampled training data. 
Some of the supervised approaches include deep cascade networks~\cite{schlemper2017deep}, CRNN~\cite{qin2018convolutional}, 
low-rank-inspired networks such as SLR-Net~\cite{ke2021learned} and L+S-Net~\cite{huang2021deep}. 
Recent dual-domain recurrent methods such as CTFNet and MostNet further exploit complementary temporal representations with data-consistency updates for dynamic MRI reconstruction~\cite{qin2021complementary, gautam2026scan}.
Diffusion-based generative models have also been investigated for dynamic MRI reconstruction~\cite{guan2025zero}. 

However, these data-driven methods typically require large, high-quality training datasets and may generalize poorly under shifts in sampling pattern, scanner protocol, anatomy, or patient population~\cite{knoll2020deep}. 
This has motivated training-data free neural reconstruction methods, including implicit neural representation (INR) approaches for dynamic cardiac MRI~\cite{kunz2024implicit, huang2025subspace} and deep image prior (DIP) methods that optimize untrained networks directly for each scan~\cite{ulyanov2018deep}. 
Recent dynamic DIP methods further incorporate temporal information through multi-dynamic DIP~\cite{vornehm2025multi} and low-rank DIP formulations~\cite{hamilton2023low}.

To address these limitations, several works incorporate DIP into optimization frameworks such as ADMM~\cite{cascarano2021combining,cheng2023mcdip,wu2023vdip} to combine the spectral bias of DIP with structured regularizations. 
While effective in practice, the theoretical behavior of such DIP-based ADMM schemes remains largely unclear due to the non-convexity of neural parameterizations. Moreover, solving the DIP subproblems typically requires multiple inner optimization steps, leading to increased computational cost. 

In this work, we propose a structured DIP framework for dynamic MRI reconstruction that explicitly models spatiotemporal dependencies through a low-rank plus sparse decomposition. Instead of directly optimizing image sequences, we parameterize both the low-rank and sparse components using separate untrained neural networks. Motivated by advances in accelerated methods~\cite{wu2024extrapolated, huang2024edge} and to address the limitations of DIP-based ADMM discussed in the previous paragraph, we adopt an extrapolated ADMM (eADMM) algorithm to improve efficiency. Empirically, extrapolation leads to faster convergence, which is further supported by our theoretical analysis. Our contributions are summarized as follows:

\begin{itemize}
    
\item \textbf{L+S Dual DIP Parametrization \& ADMM-based Algorithm}: We propose a dual DIP network re-parameterization of the dynamic MRI reconstruction problem through a low-rank plus sparse decomposition, adopting specialized regularizations for spatiotemporal dependencies. We adopt the eADMM algorithm for jointly optimizing the two networks' parameters that explicitly separates slowly varying background dynamics from sparse temporal innovations. 

\item \textbf{eADMM Convergence Analysis}: 
Under mild conditions, we provide the first convergence analysis for an eADMM framework with DIP-based generators. We establish a Lyapunov-based sufficient descent property and show that every cluster point of the generated sequence is a critical point of the proposed objective. 

\item \textbf{Extensive Empirical Evaluation}: We evaluate our L+S DIP method across different acquisition settings on the public OCMR dataset, demonstrating consistent improvements over a classical compressed sensing method, two Implicit Neural Representation (INR) methods, and two related DIP approaches. Furthermore, we show competitive performance when compared with a strong supervised learning-based baseline, \textit{all without requiring any fully-sampled training data}. 

\end{itemize}

%% file: sections/Preliminaries.tex
\section{Preliminaries}

\subsection{Dynamic MRI reconstruction \& its low-rank plus sparse modeling}

In dynamic multi-coil MRI reconstruction, the goal is to recover a sequence of discretized MR images
$\{\mathbf{x}_t\}_{t=1}^{T}$ with $\mathbf{x}_t \in \mathbb{C}^{H \times W}$ denoting the image at temporal frame $t$, from the undersampled $k$-space measurements $\{\mathbf{y}_t\}_{t=1}^{T}$. Here, $T$ is the number of temporal frames and $(H,W)$ denotes the spatial image dimensions. For each frame, the measurement model is given by
\begin{equation}
    \mathbf{y}_t = \mathbf{A}\mathbf{x}_t + \mathbf{n}_t,
\end{equation}
where, \textcolor{black}{following the common practice in dynamic MRI \cite{gamper2008compressed}}, $\mathbf{A} = \mathbf{M}\mathbf{F}\mathbf{C}$ denotes the multi-coil MRI forward operator with $\mathbf{M}$, $\mathbf{F}$, $\mathbf{C}$ denoting the undersampling mask, the 2D spatial Fourier transform and the coil sensitivity matrix, respectively, and $\mathbf{n}_t$ denotes additive measurement noise (e.g., assumed sampled from the Gaussian distribution).
To exploit temporal correlations across frames, the dynamic image sequence is commonly represented using the Casorati matrix~\cite{maximon2016differential} as follows:
\begin{equation}
    \mathbf{X}
    =
    \left[
    \mathrm{vec}(\mathbf{x}_1),
    \mathrm{vec}(\mathbf{x}_2),
    \ldots,
    \mathrm{vec}(\mathbf{x}_T)
    \right]
    \in \mathbb{C}^{HW \times T},
\end{equation}
where each column corresponds to a vectorized image frame.
Dynamic cardiac MRI data can be acquired using different protocols, including breath-held cine imaging, free-breathing acquisitions, and gated acquisitions. In this work, we focus on cardiac cine MRI reconstruction, where each slice is represented as a temporal sequence of cardiac phases.

\paragraph{Low-rank plus sparse modeling.}
Dynamic MRI sequences often exhibit strong temporal redundancy and can be decomposed into a slowly varying background (low-rank) component and a sparse dynamic component, as in~\cite {otazo2015low}: $\mathbf{X} = \mathbf{L} + \mathbf{S}$.
Here, $\mathbf{L}$ captures the temporally correlated low-rank background structure, while $\mathbf{S}$ represents sparse innovations corresponding to dynamic changes across frames. A classical low-rank plus sparse reconstruction method for dynamic MRI~\cite{otazo2015low} solves
\begin{equation}
\min_{\mathbf{L},\mathbf{S}}
\frac{1}{2}
\sum_{t=1}^{T}
\left\|
\mathbf{y}_t
-
\mathbf{A}(\mathbf{l}_t+\mathbf{s}_t)
\right\|_2^2
+
\lambda_L \|\mathbf{L}\|_*
+
\lambda_S \|\mathbf{W}\mathbf{S}\|_1,
\end{equation} 
where $\mathbf{l}_t$ and $\mathbf{s}_t$ denote the $t^{\textrm{th}}$ frame corresponding to $\mathbf{L}$ and $\mathbf{S}$, respectively. Here, $\|\cdot\|_*$ and $\|\cdot\|_1$ denote the nuclear and the $\ell_1$ norm, respectively, $\mathbf{W}$ is a sparsifying transform, and $\lambda_L,\lambda_S \geq 0$ are regularization parameters.


\subsection{Deep Image Prior for Inverse Problems}

Deep Image Prior (DIP)~\cite{ulyanov2018deep} provides a training-data-free approach for solving inverse imaging problems by parameterizing the unknown image using an untrained convolutional neural network (CNN). The architecture of the CNN, such as U-Net~\cite{ronneberger2015u}, acts as an implicit prior that favors natural image structure during optimization. CNNs show a spectral bias toward learning low frequencies (where image energy is concentrated) more quickly compared to higher frequencies (so noise fitting happens slower).
This makes DIP attractive for medical imaging settings where fully sampled training data may be limited~\cite{10896571}.
For a static inverse problem with forward model $\mathbf{y} \approx \mathbf{A}\mathbf{x}$, DIP represents the unknown image as $\mathbf{x} = f_{\bm{\theta}}(\mathbf{z})$, where $\mathbf{z}$ is a fixed random input and $\bm{\theta}$ denotes the CNN parameters. The reconstruction is obtained by fitting the network output directly to the observed measurements:
\begin{equation}
    \min_{\bm{\theta}}
    \left\|
    \mathbf{y}
    -
    \mathbf{A} f_{\bm{\theta}}(\mathbf{z})
    \right\|_2^2.
\end{equation}

A main challenge of vanilla DIP is its tendency to overfit the degraded measurements or undesired components in the null space of the forward operator, especially in the presence of noise or severe undersampling~\cite{10896571}. Several approaches have therefore been proposed to improve DIP-based reconstruction in static imaging, including early stopping, explicit regularization, stochastic perturbations, and network re-parameterizations, as recently surveyed in~\cite{11480054}.

DIP has also been extended to dynamic MRI reconstruction, beginning with time-dependent DIP methods~\cite{yoo2021timedependent} and more recent approaches such as multi-dynamic DIP~\cite{vornehm2025multi} and low-rank DIP~\cite{hamilton2023low}. While these methods improve upon DIP methods for static inverse problems by incorporating temporal information, they do not explicitly combine DIP parameterization with the classical low-rank plus sparse decomposition. 
As a result, they may not fully separate slowly varying background structure from sparse dynamic innovations, which motivates the structured dual-DIP formulation proposed in this work.

%% file: sections/methods.tex
\section{Proposed L+S DIP for dynamic MRI}

\subsection{Dual DIP re-parameterization \& the extrapolated ADMM algorithm}
Since dynamic MR image sequences exhibit strong temporal redundancy, the Casorati matrix
$\mathbf{X}$ can be decomposed into the sum of a slowly varying low-rank component and a sparse dynamic
component. Specifically, we write
\[
\mathbf{X}=\mathbf{L}+\mathbf{S},
\]
where
\[
\mathbf{L}
=
[\operatorname{vec}(\mathbf{l}_1),\operatorname{vec}(\mathbf{l}_2),\ldots,
\operatorname{vec}(\mathbf{l}_T)]
\]
denotes the low-rank component capturing the temporally correlated background structure, and
\[
\mathbf{S}
=
[\operatorname{vec}(\mathbf{s}_1),\operatorname{vec}(\mathbf{s}_2),\ldots,
\operatorname{vec}(\mathbf{s}_T)]
\]
denotes the sparse component capturing dynamic innovations across frames. Equivalently, for each
temporal frame $t$, we have
\[
\mathbf{x}_t=\mathbf{l}_t+\mathbf{s}_t .
\]

Instead of directly optimizing $\mathbf{L}$ and $\mathbf{S}$, in this paper, we propose to re-parameterize 
them using untrained CNNs, following the 
Deep Image Prior (DIP) framework~\cite{ulyanov2017deep}. 
Specifically, we introduce generator networks
$\mathcal{H}_{\bm{\theta}_L}$ and  $\mathcal{H}_{\bm{\theta}_S}$ with 3D input/output encoding for the low-rank and sparse components, respectively.
The Casorati matrices for the low rank and the sparse components can thus be defined as
\begin{equation}
\mathbf{L}(\bm{\theta}_L,\mathbf{z}_L)
=
\mathrm{reshape}\!\left(
\mathcal{H}_{\bm{\theta}_L}(\mathbf{z}_L)
\right),
\qquad
\mathbf{S}(\bm{\theta}_S,\mathbf{z}_S)
=
\mathrm{reshape}\!\left(
\mathcal{H}_{\bm{\theta}_S}(\mathbf{z}_S)
\right).
\end{equation}
where $\bm{\theta}_L$ and $\bm{\theta}_S$ denote the network parameters, respectively, and 
$\mathbf{z}_{L}$ and $\mathbf{z}_{S}$ are the 3D latent vectors representing the input to the DIP networks.
Here, $\mathcal{H}_{\bm{\theta}_L}$ and $\mathcal{H}_{\bm{\theta}_S}$ denote the DIP networks producing 3D spatiotemporal outputs, which are reshaped into Casorati matrices $\mathbf{L}$ and $\mathbf{S}$, respectively.
With the dual DIP parameterization, the optimization problem becomes:
\begin{equation}\label{model}
\min_{\bm{\theta}_L,\bm{\theta}_S}
\frac{1}{2}  \|\mathbf{Y} - \mathbf{A}(\mathbf{L}(\bm{\theta}_L,\mathbf{z}_L) + \mathbf{S}(\bm{\theta}_S,\mathbf{z}_S))\|_2^2
+ \lambda_L \|\mathbf{L}(\bm{\theta}_L,\mathbf{z}_L)\|_*
+ \lambda_S \| \mathbf{W} \mathbf{S}(\bm{\theta}_S,\mathbf{z}_S)\|_1,
\end{equation}
where $\mathbf{W}$ is a sparsifying transform such as the discrete wavelet transform or the identity transform. Here 
$
\mathbf{Y}
=
\left[
\mathrm{vec}(\mathbf{y}_1),
\mathrm{vec}(\mathbf{y}_2),
\ldots,
\mathrm{vec}(\mathbf{y}_{T})
\right]$
denotes the measurement Casorati matrix, obtained by vectorizing the $k$-space measurements from each temporal frame and stacking them as columns.
In this paper, we propose to solve the proposed optimization problem in \eqref{model} by using the extrapolated alternating direction method of multipliers (eADMM) algorithm \cite{wu2024extrapolated}. To this end, the optimization problem from Eq.~\eqref{model} can be equivalently reformulated as follows:
\begin{equation}\label{eqn: constrained opt}
\begin{aligned}
\min_{\mathbf{V},\mathbf{U},\bm{\theta}_L,\bm{\theta}_S}~
& \frac12 \|\mathbf{Y} - \mathbf{A}(\mathbf{U} + \mathbf{V})\|_2^2
+ \lambda_L\|\mathbf V\|_*
+ \lambda_S\|\mathbf{W} \mathbf U\|_1 \\
\text{s.t.}\quad
& \mathbf V=\mathbf L(\bm{\theta}_L,\mathbf{z}_L),\qquad
\mathbf U=\mathbf S(\bm{\theta}_S,\mathbf{z}_S), 
\end{aligned}
\end{equation}
where the augmented Lagrange function is
\begin{equation}
\begin{aligned}
\mathcal{L}_{\rho_L,\rho_S}
(\mathbf{V},\mathbf{U},\bm{\theta}_L,\bm{\theta}_S;\mathbf{D}_L,\mathbf{D}_S)
=&\;
\frac{1}{2}\|\mathbf{Y}-\mathbf{A}(\mathbf{U} + \mathbf{V})\|_2^2
+ \lambda_L \|\mathbf{V}\|_*
+ \lambda_S \|\mathbf{W}\mathbf{U}\|_1 \\
&+ \frac{\rho_L}{2}\|\mathbf{V} - \mathbf{L}(\bm{\theta}_L,\mathbf{z}_L)+\mathbf{D}_L\|_F^2
- \frac{\rho_L}{2}\|\mathbf{D}_L\|_F^2 \\
&+ \frac{\rho_S}{2}\|\mathbf{U} - \mathbf{S}(\bm{\theta}_S,\mathbf{z}_S)+\mathbf{D}_S\|_F^2
- \frac{\rho_S}{2}\|\mathbf{D}_S\|_F^2\:.
\end{aligned}
\label{eq:aug_lag}
\end{equation}
%
Here, $\mathbf{D}_L$ and $\mathbf{D}_S$ are the dual variables, $\rho_L$ and $\rho_S$ are augmentation penalty parameters. Algorithm \ref{alg:eadmm} provides the detailed procedure. In Step~1, we initialize the optimization variables. 
Let $\mathbf{X}^{0}=[\mathbf{x}_1^0, \mathbf{x}_2^0, \ldots, \mathbf{x}_T^0]$, where $\mathbf{x}_t^0$ denotes the \(t\)-th frame of the adjoint reconstruction $\mathbf{X}^{0}$, and the superscript $0$ denotes the initial time step of the proposed eADMM algorithm. 
We compute \(\mathbf{X}^{0}=\mathbf{A}^{H}{\mathbf{Y}}\). 
The initial low-rank component \(\mathbf{L}^{0}\) is set to be the temporal mean image \(\bar{\mathbf{x}}^{0}=\frac{1}{T}\sum_{t=1}^{T}\mathbf{x}_t^0\) repeated over frames, and the initial sparse component is defined as \(\mathbf{S}^{0}=\mathbf{X}^{0}-\mathbf{L}^{0}\). 
We then set \(\mathbf{V}^{0}=\mathbf{L}^{0}\), \(\mathbf{U}^{0}=\mathbf{S}^{0}\), \(\mathbf{D}_{L}^{0}=\mathbf{0}\), and \(\mathbf{D}_{S}^{0}=\mathbf{0}\). 
The network parameters \(\bm{\theta}_{L}^{0}\) and \(\bm{\theta}_{S}^{0}\) are randomly initialized, while the latent inputs are initialized as \(\mathbf{z}_{L}^{0},\mathbf{z}_{S}^{0}\sim\mathcal{N}(\mathbf{0}, \sigma_z^2\mathbf{I})\). 
Here, $\sigma_z$ denotes the standard deviation of the latent initialization. 

\begin{figure}[ht]
    \centering
\includegraphics[width=0.9\linewidth]{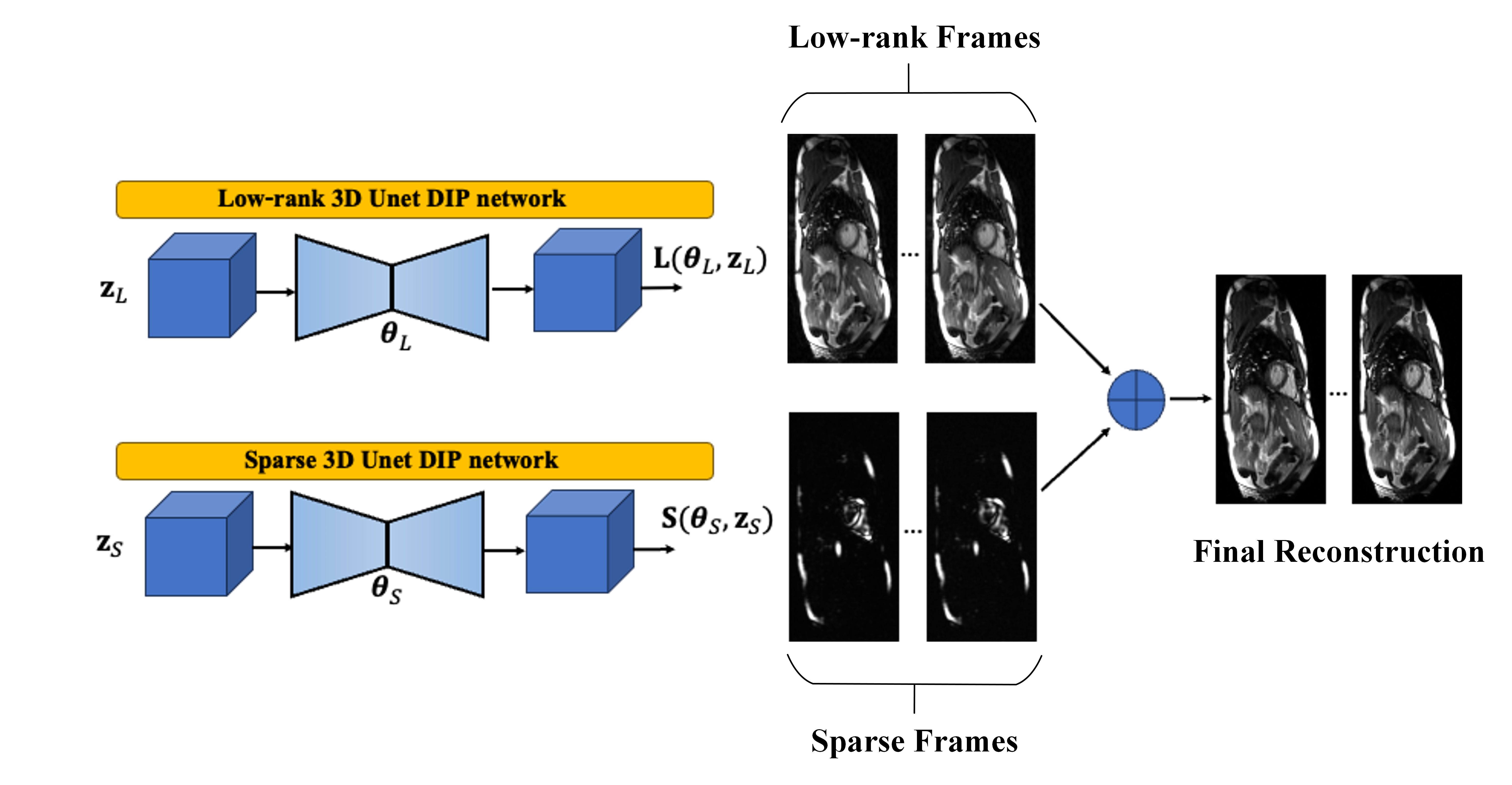}
    \caption{
Overview of the proposed L+S DIP reconstruction framework. The dynamic MR image sequence is decomposed into a low-rank component $\mathbf{L}$ and a sparse component $\mathbf{S}$, each parameterized by a separate untrained DIP network. The network outputs are combined as $\mathbf{X}=\mathbf{L}+\mathbf{S}$ and optimized using the MRI data-consistency term together with low-rank and sparsity regularization within the proposed eADMM algorithm (Eq.~\eqref{model}).
}
\label{fig:BD}
\end{figure}

The algorithm is run for $K$ optimization iterations. In Steps~5 and 7, the extrapolated variables are obtained. In Step~6, the $\mathbf{V}$-subproblem is solved via singular value thresholding \cite{cai2010singular}, whereas the $\mathbf{U}$-subproblem admits a closed-form proximal update with soft-thresholding. This is followed by solving the DIP network update subproblems by performing several steps of backpropagation using an appropriate optimizer (e.g., the Adam optimizer \cite{kingma2014adam}). An illustrative block diagram of our proposed method is given in Figure~\ref{fig:BD}.

\begin{algorithm}[t]
\caption{Extrapolated ADMM for our method}
\label{alg:eadmm}
\begin{algorithmic}[1]
\REQUIRE Measurements $\{\mathbf{y}_t,\mathbf{A}\}_{t=1}^T$, regularization parameters $\lambda_L,\lambda_S$, penalty parameters $\rho_L,\rho_S$, extrapolation parameters $\alpha,\beta$.
\STATE Initialize $\mathbf{V}^0,\mathbf{U}^0,\mathbf{D}_L^0,\mathbf{D}_S^0,\bm{\theta}_L^0,\bm{\theta}_S^0,\mathbf{z}_L,\mathbf{z}_S$
\STATE Set $\mathbf{V}^{-1}=\mathbf{V}^0$, $\mathbf{U}^{-1}=\mathbf{U}^0$
\FOR{$k=0,1,\dots,K-1$}
    \STATE Compute current network outputs
    $
    \mathbf{L}^k = \mathbf{L}(\bm{\theta}_L^k,\mathbf{z}_L),
    \mathbf{S}^k = \mathbf{S}(\bm{\theta}_S^k,\mathbf{z}_S)
    $
    \STATE Extrapolate sparse variable
    $
    \bar{\mathbf{U}}^k = \mathbf{U}^k + \beta(\mathbf{U}^k-\mathbf{U}^{k-1})
    $
    \STATE Update low-rank auxiliary variable 
    \[
    \mathbf{V}^{k+1}
    =
    \arg\min_{\mathbf{V}}
    \frac{1}{2} \|\mathbf{Y} - \mathbf{A}(\mathbf{V}+\bar{\mathbf{U}}^k)\|_2^2
    + \lambda_L \|\mathbf{V}\|_*
    + \frac{\rho_L}{2}\|\mathbf{V}-\mathbf{L}^k+\mathbf{D}_L^k\|_F^2
    \]

    \STATE Extrapolate low-rank variable
    $
    \bar{\mathbf{V}}^{k+1}
    =
    \mathbf{V}^{k+1} + \alpha(\mathbf{V}^{k+1}-\mathbf{V}^{k})
    $

    \STATE Update sparse auxiliary variable
    \[
    \mathbf{U}^{k+1}
    =
    \arg\min_{\mathbf{U}}
    \frac{1}{2} \|\mathbf{Y} - \mathbf{A}(\bar{\mathbf{V}}^{\,k+1}+\mathbf{U})\|_2^2
    + \lambda_S \|\mathbf{W}\mathbf{U}\|_1
    + \frac{\rho_S}{2}\|\mathbf{U}-\mathbf{S}^k+\mathbf{D}_S^k\|_F^2
    \]

    \STATE Update low-rank DIP variables
    \[
    \bm{\theta}_L^{k+1}
    =
    \arg\min_{\bm{\theta}_L}
    \frac{\rho_L}{2}\|\mathbf{L}(\bm{\theta}_L,\mathbf{z}_L)-(\mathbf{V}^{k+1}+\mathbf{D}_L^k)\|_F^2
    \]

    \STATE Update sparse DIP variables
    \[
    \bm{\theta}_S^{k+1}
    =
    \arg\min_{\bm{\theta}_S}
    \frac{\rho_S}{2}\|\mathbf{S}(\bm{\theta}_S,\mathbf{z}_S)-(\mathbf{U}^{k+1}+\mathbf{D}_S^k)\|_F^2
    \]

    \STATE Update
    $
    \mathbf{D}_L^{k+1}
    =
    \mathbf{D}_L^k + \mathbf{V}^{k+1} - \mathbf{L}(\bm{\theta}_L^{k+1},\mathbf{z}_L), 
    $
    $
    \mathbf{D}_S^{k+1}
    =
    \mathbf{D}_S^k + \mathbf{U}^{k+1} - \mathbf{S}(\bm{\theta}_S^{k+1},\mathbf{z}_S)
    $
\ENDFOR
\STATE \textbf{Return} $\mathbf{L}(\bm{\theta}_L^K,\mathbf{z}_L)$ and $\mathbf{S}(\bm{\theta}_S^K,\mathbf{z}_S)$
\end{algorithmic}
\end{algorithm}

\subsection{Theoretical convergence analysis}
Existing DIP methods that utilize ADMM mainly focus on algorithmic design and empirical performance \cite{bell2025tada,cheng2023mcdip}, while theoretical convergence analysis remains largely unexplored. In contrast, in this paper, we provide a systematic convergence analysis for extrapolated ADMM with the two DIP-based priors. 
We analyze the convergence of the proposed extrapolated ADMM in Algorithm~\ref{alg:eadmm} for solving the nonconvex and non-smooth (due to the $\ell_1$ norm regularization) objective function \eqref{model}. 
Due to space constraints, the proofs are deferred to Appendix \ref{app:th}.

By introducing auxiliary variables $\mathbf{V}$ and $\mathbf{U}$, the model \eqref{model} can be equivalently reformulated as the constrained optimization problem in \eqref{eqn: constrained opt}. 
The functions 
\[
g(\mathbf V)=\lambda_L\|\mathbf V\|_*\:, \quad \text{and}  \quad
h(\mathbf U)=\lambda_S\|\mathbf{W} \mathbf U\|_1
\]
are proper lower semicontinuous. 
Moreover, the data-fidelity term 
\[
f(\mathbf V,\mathbf U)=\frac12 \|\mathbf{Y} - \mathbf{A}(\mathbf{U} + \mathbf{V})\|_2^2
\]
is continuously differentiable with globally Lipschitz gradient since it is quadratic in $(\mathbf V,\mathbf U)$.
We now establish the convergence of Algorithm \ref{alg:eadmm} by stating some standard assumptions for \eqref{model}.

\begin{assumption}\label{mass2:Abounded}
The linear operators $\{\mathbf{A}_t\}_{t=1}^T$ and $\mathbf{W}$ are bounded. The sequence
\[
\{(\mathbf V^k,\mathbf U^k,\bm{\theta}_L^k,\bm{\theta}_S^k;\mathbf D_L^k,\mathbf D_S^k)\}_{k\ge0}
\] 
generated by Algorithm~\ref{alg:eadmm} is bounded.
\end{assumption}

\begin{assumption}\label{mass3:inexact}
There exist nonnegative sequences $\{\varepsilon_k^L\}$ and $\{\varepsilon_k^S\}$ such that
$
\sum_{k=0}^{\infty}\varepsilon_k^L<\infty$, 
$\sum_{k=0}^{\infty}\varepsilon_k^S<\infty,
$
and
\begin{equation}\label{DIPdescent}
\begin{aligned}
&\frac{\rho_L}{2}\|\mathbf L(\bm{\theta}_L^{k+1},\mathbf{z}_L)-(\mathbf V^{k+1}+\mathbf D_L^k)\|_F^2
\leq
\frac{\rho_L}{2}\|\mathbf L(\bm{\theta}_L^k,\mathbf{z}_L)-(\mathbf V^{k+1}+\mathbf D_L^k)\|_F^2
+\varepsilon_k^L,\\
&\frac{\rho_S}{2}\|\mathbf S(\bm{\theta}_S^{k+1},\mathbf{z}_S)-(\mathbf U^{k+1}+\mathbf D_S^k)\|_F^2
\leq
\frac{\rho_S}{2}\|\mathbf S(\bm{\theta}_S^{k},\mathbf{z}_S)-(\mathbf U^{k+1}+\mathbf D_S^k)\|_F^2+\varepsilon_k^S
\end{aligned}
\end{equation}
\end{assumption} 
For brevity, we denote
$
\mathcal L^{k}
:=
\mathcal L_{\rho_L,\rho_S}
(\mathbf V^{k},\mathbf U^{k},\bm{\theta}_L^{k},\bm{\theta}_S^{k};\mathbf D_L^{k},\mathbf D_S^{k}).
$ 
To handle the extrapolation terms, we introduce the following Lyapunov function $\Psi^k $ \cite{taylor2018lyapunov} 
with the augmented Lagrange function
\begin{equation}\label{eq:lyapunov}
\Psi^k :=
\mathcal L^{k}
+\frac{\delta_V}{2}\|\mathbf V^k-\mathbf V^{k-1}\|_F^2
+\frac{\delta_U}{2}\|\mathbf U^k-\mathbf U^{k-1}\|_F^2,
\end{equation}
where $\delta_V>0$ and $\delta_U>0$ are constants.
To establish the sufficient descent property, the extrapolation-induced error terms need to be dominated by the strongly convex parts of the $\mathbf V$- and $\mathbf U$-subproblems. 
This leads to the following parameter condition.
\begin{assumption}\label{mass:param}
The Lyapunov parameters $\delta_V,\delta_U$ are chosen such that
\begin{equation}
\delta_V < \tau_V - 2c_1, \qquad
2c_3 < \delta_U < \tau_U - 2c_2,
\end{equation}
where $\tau_V$ and $\tau_U$ are positive strong convexity moduli of the $\mathbf V$- and $\mathbf U$-subproblems, and $c_1:= 2\alpha^2L_f+\frac{\alpha L_f}{2}\eta_1+\frac{\beta L_f}{2}\eta_2$, $c_2:=\frac{\alpha L_f}{2\eta_1}$, and $c_3:=2\beta^2L_f+\frac{\beta L_f}{2\eta_2}$, 
depend on the extrapolation parameters $\alpha,\beta$, the Lipschitz constant $L_f$, and positive constants $\eta_1, \eta_2$.
\end{assumption}

\begin{lemma}\label{mlem:descent}
Suppose Assumptions~\ref{mass2:Abounded}, \ref{mass3:inexact}, and \ref{mass:param} hold. Then, there exist constants $\xi_V,\xi_U>0$ and a positive summable sequence $\{\gamma_k\}$ such that
\begin{equation}
\Psi^{k+1} - \Psi^k
\le
-\xi_V\|\mathbf V^{k+1}-\mathbf V^k\|_F^2
-\xi_U\|\mathbf U^{k+1}-\mathbf U^k\|_F^2
+\gamma_k.
\end{equation}
\end{lemma}
\begin{remark}\label{re:epar}
The sufficient descent condition requires 
Assumption \ref{mass:param}. 
Hence, there exists a nonempty neighborhood of $(\alpha,\beta)$ such that the sufficient descent property holds for all extrapolation parameters in this neighborhood. 
In particular, a sufficient choice is
\begin{equation}\label{eqn: extrapolation}
0\le \alpha<\frac{\eta_1(\tau_U-\delta_U)}{L_f},
\qquad
0\le \beta<\beta_{\max},
\end{equation}
where $\beta_{\max}$ is the positive root of $2L_f\beta^2+\frac{L_f}{2\eta_2}\beta-\frac{\delta_U}{2}=0$.
\end{remark}

As a direct consequence of Lemma~\ref{mlem:descent}, we have the following lemma. 
\begin{lemma}\label{lem:asymptotic}
Suppose Assumptions~\ref{mass2:Abounded}, \ref{mass3:inexact}, and \ref{mass:param} hold. Then
\begin{equation}\label{eq:square-summable}
\sum_{k=0}^{\infty}\|\mathbf V^{k+1}-\mathbf V^k\|_F^2<\infty,
\qquad
\sum_{k=0}^{\infty}\|\mathbf U^{k+1}-\mathbf U^k\|_F^2<\infty.
\end{equation}
Consequently,
\begin{equation}\label{eq:vanishing-VU}
\|\mathbf V^{k+1}-\mathbf V^k\|_F\to0,
\qquad
\|\mathbf U^{k+1}-\mathbf U^k\|_F\to0.
\end{equation}
\end{lemma}

We next show that the limiting subgradient of the Lyapunov function can be controlled by successive differences of the iterates.

\begin{lemma}\label{mlem:subgradient}
Suppose Assumptions~\ref{mass2:Abounded}, \ref{mass3:inexact}, and \ref{mass:param} hold. Denote $\mathbf{L}^{k}:=\mathbf L(\bm{\theta}_L^{k},\mathbf{z}_L)$ and $\mathbf{S}^k:=\mathbf S(\bm{\theta}_S^{k},\mathbf{z}_S)$, then there exists a constant $C>0$ and a summable sequence $\{\gamma_k\}$ such that
\begin{equation}
\begin{aligned}
\mathrm{dist}(0,\partial \Psi^{k+1})
\le\;&
C\Big(
\|\mathbf V^{k+1}-\mathbf V^k\|_F
+\|\mathbf U^{k+1}-\mathbf U^k\|_F
+\|\mathbf U^k-\mathbf U^{k-1}\|_F
\\&
+\|\mathbf L^{k+1}-\mathbf L^k\|_F
+\|\mathbf S^{k+1}-\mathbf S^k\|_F
+\|\mathbf D_L^{k+1}-\mathbf D_L^k\|_F
+\|\mathbf D_S^{k+1}-\mathbf D_S^k\|_F
\Big)
+\gamma_k,
\end{aligned}
\label{eq:subgrad}
\end{equation}
where $\partial \Psi$ denotes the (limiting) subdifferential of $\Psi$, and $\mathrm{dist}(0,\partial \Psi^{k+1})$ measures how close the iterate is to satisfying the first-order optimality condition.
\end{lemma}

\begin{theorem}\label{thm:basic}
Suppose Assumptions~\ref{mass2:Abounded}, \ref{mass3:inexact}, and \ref{mass:param} hold. 
Let the sequence
$\{(\mathbf V^k,\mathbf U^k,\bm{\theta}_L^k,\bm{\theta}_S^k,\mathbf D_L^k,\mathbf D_S^k)\}$
be generated by Algorithm~\ref{alg:eadmm}. Then:
\begin{itemize}
\item[(i)] The set of cluster points of the sequence is nonempty and compact.

\item[(ii)] The sequence $\{\Psi^k\}$ converges. 

\end{itemize}
\end{theorem}
\begin{remark}\label{mass:auxiliary}
The estimate in Lemma~\ref{mlem:subgradient} involves auxiliary increments
\[
\|\mathbf L^{k+1}-\mathbf L^k\|_F,\quad
\|\mathbf S^{k+1}-\mathbf S^k\|_F,\quad
\|\mathbf D_L^{k+1}-\mathbf D_L^k\|_F,\quad
\|\mathbf D_S^{k+1}-\mathbf D_S^k\|_F.
\]
Controlling or vanishing of these quantities is sufficient to ensure that the limiting subgradient tends to zero.
\end{remark}

\begin{corollary}\label{cor:critical}
If, in addition, the auxiliary increments in Remark~\ref{mass:auxiliary} vanish asymptotically, then every cluster point of the generated sequence is a critical point of $\Psi$.
\end{corollary}

\begin{remark}
If the Lyapunov function $\Psi$ satisfies the Kurdyka--\L ojasiewicz (KL) property, and additional regularity conditions ensure sufficient control of auxiliary increments, then the standard KL framework can be applied to establish convergence of the sequence.
\end{remark}


%% file: sections/experiments.tex
\section{Experiments}

\subsection{Experimental Setup}

We evaluate the proposed L+S DIP framework on dynamic cardiac MRI reconstruction. 
We use the publicly available OCMR dataset~\cite{chen2020ocmrv10}, which provides fully sampled multi-coil cardiac cine MRI acquisitions. The dataset consists of short-axis cardiac sequences with 20--30 temporal frames per slice. In our experiments, we evaluate on 26 cine slices, with each slice preprocessed to a reconstruction matrix size of $300 \times 144$. Fully sampled $k$-space data are retrospectively undersampled to simulate accelerated acquisitions.
We used the fully sampled central $k$-space lines for estimating the sensitivity maps via ESPIRiT~\cite{uecker2014espirit}. We use a Cartesian undersampling pattern with variable-density random sampling~\cite{lustig2007sparse}. More specifically, we sample lines along the phase-encoding direction that are shared across all temporal frames of each cine sequence. 
The fully sampled $k$-space data are retrospectively undersampled to simulate accelerated acquisitions. 
We consider acceleration factors (AFs) of $4\times$, 6$\times$ and $8\times$.

We compare the proposed method against several baselines for dynamic cardiac MRI reconstruction: (i) compressed sensing~\cite{lustig2007sparse} to represent classical optimization methods; (ii) L+S-Net~\cite{huang2021deep} to represent supervised learning-based methods, which unrolls an $L+S$-inspired reconstruction framework and learns reconstruction modules from training data; (iii) Recent DIP-based training-data-free methods: Multi-dynamic DIP~\cite{vornehm2025multi} and Low-rank DIP (LR-DIP)~\cite{hamilton2023low}; and (iv) an Implicit neural representation (INR) method, 
Fourier-feature implicit neural networks for free-breathing cardiac MRI reconstruction~\cite{kunz2024implicit}, to represent non-DIP-based training-data-free method. This selection is intended to evaluate whether the proposed structured $L+S$ DIP formulation improves over both conventional optimization techniques and recent neural data-intensive and data-less approaches.

For our L+S DIP, both DIP generators are implemented as 3D U-Net. The networks are randomly initialized and optimized per scan using Adam with learning rate $3 \times 10^{-4}$. The regularization parameters $\lambda_L$ and $\lambda_S$ are selected based on the study in Table~\ref{tab:lambda_grid} of Appendix~\ref{appen: impact of choice of reg params}. The optimization is performed for \(K=3000\) eADMM outer iterations, and DIP networks' subproblems are optimized with \(N_{\mathrm{in}}=20\) Adam steps at each outer iteration. 
All experiments are conducted on a single NVIDIA RTX PRO 6000 Blackwell Server Edition GPU.

\subsection{Main results}

In Table~\ref{tab:main_results}, we report the reconstruction performance in terms of the average PSNR and SSIM. As observed, across different acceleration factors, the proposed method consistently outperforms the compressed sensing method, other DIP approaches, and INR methods, particularly at higher acceleration factors. For example, at $8\times$, our method achieves 3 dB higher PSNR than the INR-based baseline and nearly 2 dB more when compared to other DIP baselines. 
\begin{table*}[t]
\centering
\caption{
Quantitative comparison on the OCMR dataset at acceleration factors AF = 4$\times$, 6$\times$, and 8$\times$. 
PSNR and SSIM are reported, with higher values indicating better reconstruction quality. 
The proposed L+S DIP method achieves the highest PSNR and SSIM across all acceleration factors, demonstrating strong reconstruction performance without supervised training.}
\label{tab:main_results}
\resizebox{0.88\textwidth}{!}{
\begin{tabular}{lccccccc}
\toprule
\multirow{2}{*}{Method} 
& \multirow{2}{*}{Training data }  
& \multicolumn{2}{c}{AF =  $4\times$} 
& \multicolumn{2}{c}{AF =  $6\times$} 
& \multicolumn{2}{c}{AF = $8\times$} \\
\cmidrule(lr){3-4} \cmidrule(lr){5-6} \cmidrule(lr){7-8}
 & & PSNR & SSIM & PSNR & SSIM & PSNR & SSIM \\
\midrule
Compressed sensing~\cite{lustig2007sparse} & Not needed & 29.37 & 0.84 & 27.33 & 0.79 & 24.48 & 0.74 \\
\midrule
L+S-Net~\cite{huang2021deep} & Needed
& 34.55 & 0.96 & 31.97 & 0.93 & 28.32 & 0.87 \\
\midrule
Fourier-feature INR~\cite{kunz2024implicit} & Not needed & 30.59 & 0.84 & 29.33 & 0.81 & 26.64 & 0.75 \\
\midrule
Multi-dynamic DIP~\cite{vornehm2025multi} & Not needed & 30.06 & 0.79 & 29.97 & 0.77 & 26.44 & 0.71 \\
LR-DIP~\cite{hamilton2023low} & Not needed & 31.45 & 0.84 & 29.81 & 0.84 & 27.73 & 0.78 \\
\midrule
\textbf{L+S DIP (Ours)} & Not needed 
& \textbf{34.62} & \textbf{0.96} 
& \textbf{33.06} & \textbf{0.94} 
& \textbf{29.74} & \textbf{0.90} \\
\bottomrule
\end{tabular}
}
\end{table*}

Compared to the supervised learning method, L+S-Net, our method achieves marginal improvements, indicating that we obtain competitive performance, \textit{all without requiring any training data}. That said, since our method requires optimization for every sequence, the run-time required is higher when compared to the inference time of L+S-Net.



%







\subsection{Impact of the proposed dual architecture}

%


\begin{figure}
    \centering
    \includegraphics[width=0.7\linewidth]{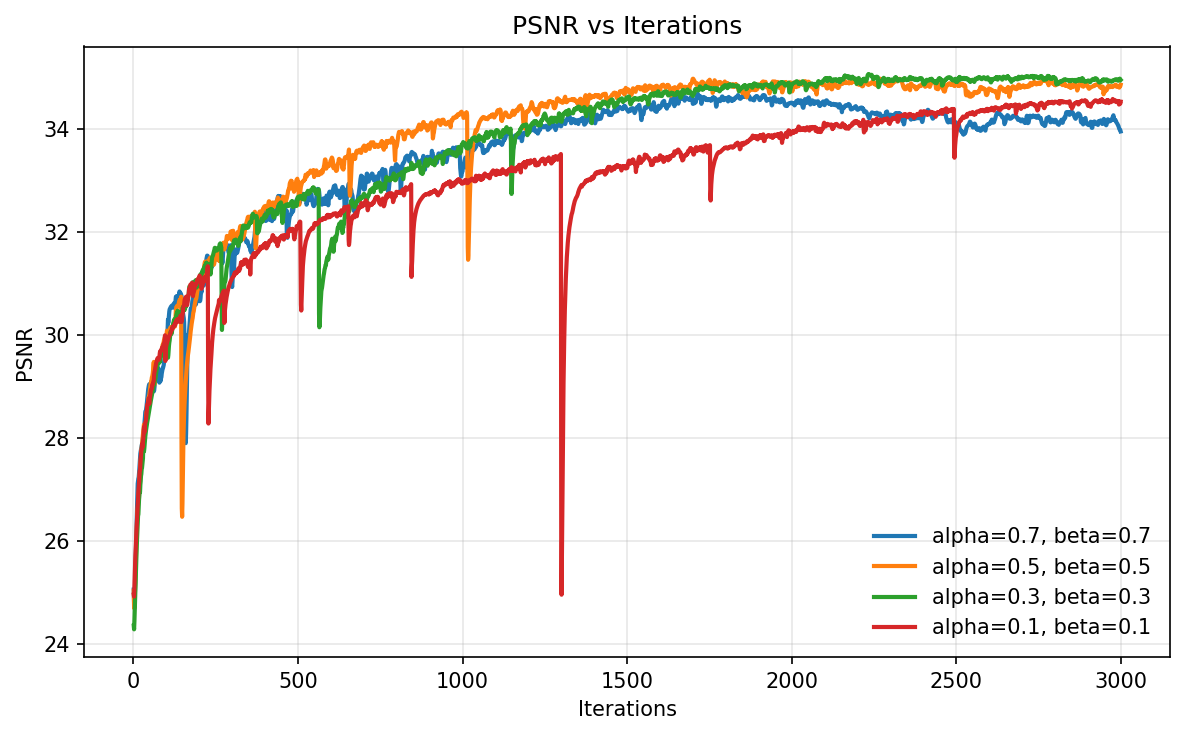}
    \vspace{-0.4cm}
    \caption{{
Effect of the extrapolation parameters $\alpha$ and $\beta$ in the proposed eADMM algorithm on the reconstruction convergence. The curves show PSNR as a function of optimization iteration for different choices of $\alpha$ and $\beta$, demonstrating that moderate extrapolation improves convergence speed.
}}
    \label{fig:alpha_beta}
\end{figure}

Table~\ref{tab:ablation} evaluates the contribution of each component of our framework. Experiments are performed on two slices from the OCMR dataset under a $4\times$ acceleration setting.
The low-rank-only variant sets $\lambda_S=0$, and the sparse-only variant sets $\lambda_L=0$. We also test a single-network DIP variant, where the low-rank and sparse components share the same network. The results show that our framework achieves the best performance, with a PSNR of 33.83 dB. 
\begin{wraptable}{r}{0.48\linewidth}
\vspace{-0.3cm}
\centering
\caption{Average PSNR and SSIM of our method versus different DIP architectures.}
\label{tab:ablation}
\small
\resizebox{0.95\linewidth}{!}{
\begin{tabular}{lcc}
\toprule
Method & PSNR & SSIM \\
\midrule
Single-network DIP ($\theta_L=\theta_S=\theta$) & 31.23 & 0.93 \\
Low-rank-only DIP ($\lambda_S=0$) & 29.82 & 0.86 \\
Sparse-only DIP ($\lambda_L=0$) & 30.18 & 0.86 \\
\midrule
\textbf{Proposed (L+S DIP)} & \textbf{33.83} & \textbf{0.96} \\
\bottomrule
\end{tabular}}
\vspace{-0.4cm}
\end{wraptable}
Our framework improves PSNR by 2.60 dB, 4.01 dB, and 3.65 dB over the single-network, low-rank-only, and sparse-only variants respectively. These results show that the low-rank and sparse components are complementary. The low-rank component models the shared anatomical structure across time, while the sparse component captures temporal variations and residual dynamic details.

Therefore, decomposing the reconstruction into low-rank and sparse components is important for our framework.


\subsection{A study on the acceleration achieved via the extrapolation parameters in Algorithm~\ref{alg:eadmm}}






In this subsection, we study the effect of the extrapolation parameters $\alpha$ and $\beta$ in Algorithm~\ref{alg:eadmm}. 
Figure~\ref{fig:alpha_beta} shows the PSNR curves under different choices of $\alpha$ and $\beta$ as defined in Remark~\ref{re:epar}. 
All experiments are conducted on a single slice from the OCMR dataset with a $4\times$ acceleration factor.
The joint extrapolation setting with $\alpha=\beta = 0.5$ achieves the best performance and convergence, as compared with the extrapolated baseline with the lowest values of $\alpha$ and $\beta$, and the highest values.  
These results suggest that moderate extrapolation can improve the convergence speed of Algorithm~\ref{alg:eadmm}. However, overly large extrapolation parameters may introduce instability.

\subsection{Qualitative results}

\begin{figure}[ht]
\centering
\includegraphics[width=0.9\textwidth]{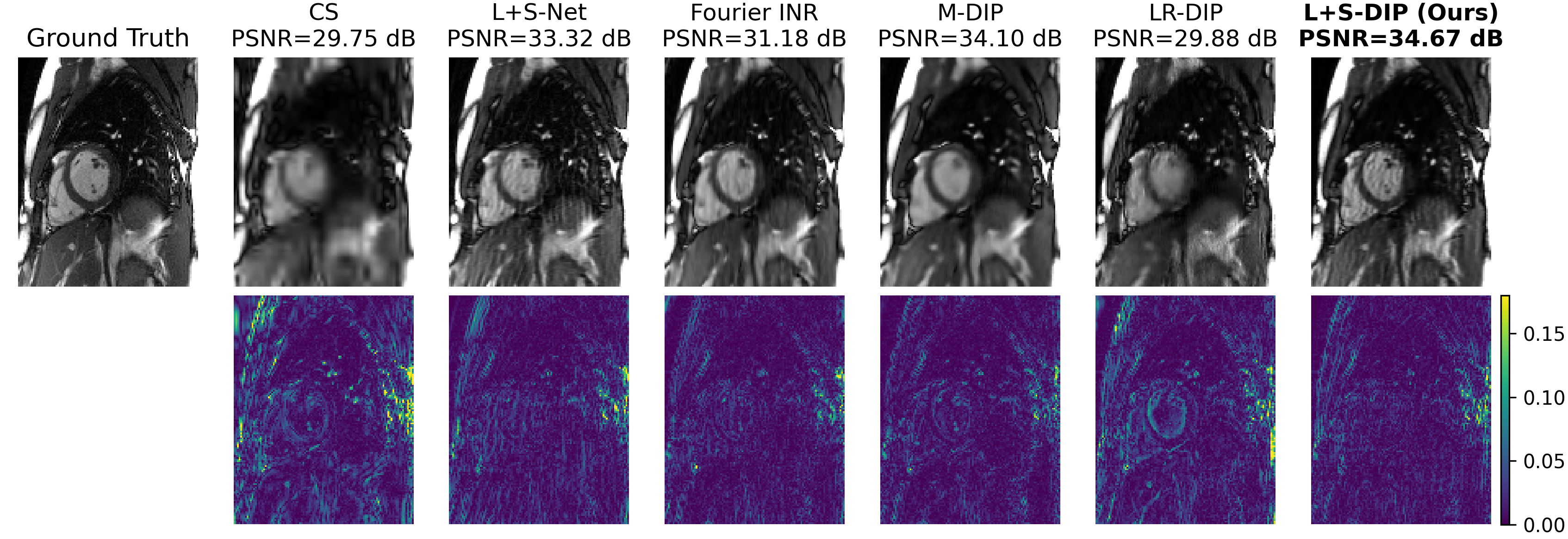}
\vspace{-0.3cm}
\caption{
Comparison of the proposed L+S-DIP framework with other dynamic MRI reconstruction methods on the OCMR dataset under 4$\times$ acceleration. 
The top row shows the ground truth image followed by the reconstructed images, and the bottom row shows the corresponding error maps. 
Values above each reconstruction indicate the PSNR (in dB) over the full image.
The proposed L+S-DIP method better preserves cardiac structure and reduces residual errors around dynamic regions compared with the baseline methods.}
\label{fig:qualitative_results}
\end{figure}

Figure~\ref{fig:qualitative_results} shows qualitative reconstruction results on the OCMR dataset. 
For each method, we show the reconstructed ROI image and the corresponding absolute error map with respect to the ground truth for the same temporal frame. 
The PSNR values reported above the reconstructions provide a quantitative measure of image fidelity for the displayed frame.
Compressed sensing reconstructs the main cardiac structure, but exhibits noticeable blurring and residual artifacts. 
L+S-Net improves the reconstruction quality relative to Compressed Sensing, although some fine anatomical details remain oversmoothed. 
Fourier-feature INR and LR-DIP show visible residual errors around the heart and surrounding dynamic regions. 
Multi-dynamic DIP produces smoother reconstructions, but still loses some fine structural details.

In contrast, the proposed L+S DIP method produces a reconstruction that is visually closer to the ground truth and achieves the highest PSNR among the compared methods for the displayed frame. 
Its error map shows reduced residual error around the cardiac region and surrounding tissues, indicating improved preservation of anatomical structure in dynamic MRI reconstruction.
In Appendix~\ref{appen: uncertainty maps}, we provide uncertainty maps when using different initializations for our L+S DIP method. In Appendix~\ref{appen: low rank and sparse visuals}, we present visualizations on the learned low-rank and sparse components of L+S DIP.

%% file: sections/discussions.tex
\section{Conclusion}
This work proposed a training-free dynamic MRI reconstruction framework that combines deep image priors with explicit low-rank and sparse modeling. By integrating the implicit bias of DIP with structured spatiotemporal regularization, the proposed L+S-DIP approach effectively captures spatiotemporal redundancy in dynamic MRI while avoiding the need for supervised training data. This makes the framework particularly attractive in medical imaging settings where fully sampled datasets are limited. 
Experimental results on the OCMR dataset demonstrate competitive reconstruction quality across multiple acceleration factors, consistently achieving strong PSNR and SSIM performance while preserving fine anatomical structures and reducing reconstruction artifacts. A discussion of the limitations of the proposed approach is provided in Appendix~\ref{appen: limitations}.

%% file: sections/appendix.tex
\newpage
\appendix
\onecolumn
\par\noindent\rule{\textwidth}{1pt}
\begin{center}
{\Large \bf Appendix}
\end{center}
\vspace{-0.1in}
\par\noindent\rule{\textwidth}{1pt}
\appendix

\section{Theoretical convergence analysis}\label{app:th}
To ensure the completeness and readability of the proof, we will restate the problem and the assumptions.

We analyze the convergence of the proposed extrapolated ADMM Algorithm~\ref{alg:eadmm} for solving the nonconvex and nonsmooth problem. 
By introducing auxiliary variables $\mathbf{V}$ and $\mathbf{U}$, the model \eqref{model} can be equivalently reformulated as the constrained optimization problem
\begin{equation} 
\begin{aligned}
\min_{\mathbf{V},\mathbf{U},\bm{\theta}_L,\bm{\theta}_S}~
& \frac12 \|\mathbf{Y} - \mathbf{A}(\mathbf{U} + \mathbf{V})\|_2^2
+ \lambda_L\|\mathbf V\|_*
+ \lambda_S\|\mathbf{W} \mathbf U\|_1 \\
\text{s.t.}\quad
& \mathbf V=\mathbf L(\bm{\theta}_L,\mathbf{z}_L),\qquad
\mathbf U=\mathbf S(\bm{\theta}_S,\mathbf{z}_S), 
\end{aligned}
\end{equation}
The functions $g(\mathbf V)=\lambda_L\|\mathbf V\|_*$ and 
$h(\mathbf U)=\lambda_S\|\mathbf W \mathbf U\|_1$ are proper lower semicontinuous. 
Moreover, the data-fidelity term
\[
f(\mathbf V,\mathbf U) = \frac12 \|\mathbf{Y} - \mathbf{A}(\mathbf{U} + \mathbf{V})\|_2^2
\]
is continuously differentiable with globally Lipschitz gradient since it is quadratic in $(\mathbf V,\mathbf U)$.

\begin{assumption}\label{ass2:Abounded}
The linear operators $\{\mathbf{A}_t\}_{t=1}^T$ and $\mathbf{W}$ are bounded. The sequence

\[
\left\{
\left(
\mathbf V^k,\mathbf U^k,
\bm{\theta}_L^k,\bm{\theta}_S^k;
\mathbf D_L^k,\mathbf D_S^k
\right)
\right\}_{k\ge0}
\]
generated by Algorithm~\ref{alg:eadmm} is bounded.
\end{assumption}

\begin{assumption}\label{ass3:inexact}
There exist nonnegative sequences $\{\varepsilon_k^L\}$ and $\{\varepsilon_k^S\}$ such that
$
\sum_{k=0}^{\infty}\varepsilon_k^L<\infty$, 
$\sum_{k=0}^{\infty}\varepsilon_k^S<\infty,
$
and
\begin{equation}
\begin{aligned}
&\frac{\rho_L}{2}\|\mathbf L(\bm{\theta}_L^{k+1},\mathbf{z}_L)-(\mathbf V^{k+1}+\mathbf D_L^k)\|_F^2
\leq
\frac{\rho_L}{2}\|\mathbf L(\bm{\theta}_L^k,\mathbf{z}_L)-(\mathbf V^{k+1}+\mathbf D_L^k)\|_F^2
+\varepsilon_k^L,\\
&\frac{\rho_S}{2}\|\mathbf S(\bm{\theta}_S^{k+1},\mathbf{z}_S)-(\mathbf U^{k+1}+\mathbf D_S^k)\|_F^2
\leq
\frac{\rho_S}{2}\|\mathbf S(\bm{\theta}_S^{k},\mathbf{z}_S)-(\mathbf U^{k+1}+\mathbf D_S^k)\|_F^2+\varepsilon_k^S
\end{aligned}
\end{equation}
\end{assumption} 
For brevity, we denote
\begin{equation}
\mathcal L^{k}
:=
\mathcal L_{\rho_L,\rho_S}
(\mathbf V^{k},\mathbf U^{k},\bm{\theta}_L^{k},\bm{\theta}_S^{k};\mathbf D_L^{k},\mathbf D_S^{k}).
\end{equation}
To handle the extrapolation terms, we introduce the following Lyapunov function $\Psi^k $ \cite{taylor2018lyapunov} 
with the augmented Lagrange function
\begin{equation}
\Psi^k :=
\mathcal L^{k}
+\frac{\delta_V}{2}\|\mathbf V^k-\mathbf V^{k-1}\|_F^2
+\frac{\delta_U}{2}\|\mathbf U^k-\mathbf U^{k-1}\|_F^2,
\end{equation}
where $\delta_V>0$ and $\delta_U>0$ are constants.
To establish the sufficient descent property, the extrapolation-induced error terms need to be dominated by the strongly convex parts of the $\mathbf V$- and $\mathbf U$-subproblems. 
This leads to the following parameter condition.
\begin{assumption}\label{ass:param}
The Lyapunov parameters $\delta_V,\delta_U$ are chosen such that
\begin{equation}
\delta_V < \tau_V - 2c_1, \qquad
2c_3 < \delta_U < \tau_U - 2c_2,
\end{equation}
where $\tau_V$, $\tau_U$ are positive strong convexity moduli of $\mathbf V$- and $\mathbf U$-subproblems, and $c_1:= 2\alpha^2L_f+\frac{\alpha L_f}{2}\eta_1+\frac{\beta L_f}{2}\eta_2$, $c_2:=\frac{\alpha L_f}{2\eta_1}$, and $c_3:=2\beta^2L_f+\frac{\beta L_f}{2\eta_2}$, 
depend on extrapolation parameters $\alpha,\beta$, the Lipschitz constant $L_f$, and positive constants $\eta_1, \eta_2$.

\end{assumption}
\begin{lemma}\label{lem:descent}
Suppose Assumptions~\ref{ass2:Abounded}, \ref{ass3:inexact} and \ref{ass:param} hold. Then there exist constants $\xi_V,\xi_U>0$ and a summable sequence $\{\gamma_k\}$ such that
\begin{equation}
\Psi^{k+1} - \Psi^k
\le
-\xi_V\|\mathbf V^{k+1}-\mathbf V^k\|_F^2
-\xi_U\|\mathbf U^{k+1}-\mathbf U^k\|_F^2
+\gamma_k.
\end{equation}
\end{lemma}

To prove Lemma \ref{lem:descent}, we first recall the descent property of a smooth function in the following lemma.
\begin{lemma}\label{descent property}{\rm\cite{bolte2014proximal}}
 Let $f:~\mathbb{R}^{n}\rightarrow \mathbb{R}$ be a continuously differentiable function with gradient $\nabla f$ assumed $L_f$-Lipschitz continuous. Then, for any $u,v\in \mathbb{R}^{n}$, we have
 \begin{equation*}
 \left|f(u)-f(v)-\langle u-v,\nabla f(v)\rangle\right|\leq \frac{L_f}{2}\|u-v\|^2.\end{equation*}
 \end{lemma}
\begin{proof}
Using the optimality of the updating schemes of $\mathbf V^{k+1}$ and $\mathbf U^{k+1}$ in Algorithm \ref{alg:eadmm}, we have 
\begin{align}
&f(\mathbf V^{k+1},\bar{\mathbf U}^k)+\lambda_L\|\mathbf V^{k+1}\|_*
+\frac{\rho_L}{2}\|\mathbf V^{k+1}-\mathbf L^k+\mathbf D_L^k\|_F^2
\notag\\
\le\;&
f(\mathbf V^k,\bar{\mathbf U}^k)+\lambda_L\|\mathbf V^k\|_*
+\frac{\rho_L}{2}\|\mathbf V^k-\mathbf L^k+\mathbf D_L^k\|_F^2.
\label{eq:V-descent}
\end{align}
and
\begin{align}
&f(\bar{\mathbf V}^{k+1},\mathbf U^{k+1})+\lambda_S\|\mathbf W\mathbf U^{k+1}\|_1
+\frac{\rho_S}{2}\|\mathbf U^{k+1}-\mathbf S^k+\mathbf D_S^k\|_F^2
\notag\\
\le\;&
f(\bar{\mathbf V}^{k+1},\mathbf U^k)+\lambda_S\|\mathbf W\mathbf U^k\|_1
+\frac{\rho_S}{2}\|\mathbf U^k-\mathbf S^k+\mathbf D_S^k\|_F^2.
\label{eq:U-descent}
\end{align}
Summing \eqref{eq:V-descent} and \eqref{eq:U-descent}, and adding and subtracting the $f(\mathbf V^{k+1},\mathbf U^{k+1})$ on the left-hand side, and
$f(\mathbf V^k,\mathbf U^k)$ on the right-hand side, we obtain 

\begin{align}
&f(\mathbf V^{k+1},\mathbf U^{k+1})+f(\mathbf V^{k+1},\bar{\mathbf U}^k)-f(\mathbf V^{k+1},\mathbf U^{k+1})
+\lambda_L\|\mathbf V^{k+1}\|_*
+\frac{\rho_L}{2}\|\mathbf V^{k+1}-\mathbf L^k+\mathbf D_L^k\|_F^2
\notag\\
&\quad
+f(\bar{\mathbf V}^{k+1},\mathbf U^{k+1})
+\lambda_S\|\mathbf W\mathbf U^{k+1}\|_1
+\frac{\rho_S}{2}\|\mathbf U^{k+1}-\mathbf S^k+\mathbf D_S^k\|_F^2
\notag\\
\le\;&
f(\mathbf V^k,\mathbf U^k)+
f(\mathbf V^{k},\bar{\mathbf U}^k)-f(\mathbf V^k,\mathbf U^k)
+\lambda_L\|\mathbf V^{k}\|_*
+\frac{\rho_L}{2}\|\mathbf V^{k}-\mathbf L^k+\mathbf D_L^k\|_F^2
\notag\\
&\quad
+f(\bar{\mathbf V}^{k+1},\mathbf U^{k})
+\lambda_S\|\mathbf W\mathbf U^{k}\|_1
+\frac{\rho_S}{2}\|\mathbf U^{k}-\mathbf S^k+\mathbf D_S^k\|_F^2 .
\end{align}
Rearrange this, we obtain
\begin{align}
&f(\mathbf V^{k+1},\mathbf U^{k+1})
+\lambda_L\|\mathbf V^{k+1}\|_*+\lambda_S\|\mathbf W\mathbf U^{k+1}\|_1
\notag\\
&\quad
+\frac{\rho_L}{2}\|\mathbf V^{k+1}-\mathbf L^k+\mathbf D_L^k\|_F^2+\frac{\rho_S}{2}\|\mathbf U^{k+1}-\mathbf S^k+\mathbf D_S^k\|_F^2
\notag\\
\le\;&
f(\mathbf V^k,\mathbf U^k)
+\lambda_L\|\mathbf V^{k}\|_*+\lambda_S\|\mathbf W\mathbf U^{k}\|_1
\notag\\
&\quad
+\frac{\rho_L}{2}\|\mathbf V^{k}-\mathbf L^k+\mathbf D_L^k\|_F^2
+\frac{\rho_S}{2}\|\mathbf U^{k}-\mathbf S^k+\mathbf D_S^k\|_F^2\\&\quad\nonumber
+f(\bar{\mathbf V}^{k+1},\mathbf U^{k})-f(\bar{\mathbf V}^{k+1},\mathbf U^{k+1})\\&\quad\nonumber
+f(\mathbf V^{k},\bar{\mathbf U}^k)-f(\mathbf V^k,\mathbf U^k)
+f(\mathbf V^{k+1},\mathbf U^{k+1})-f(\mathbf V^{k+1},\bar{\mathbf U}^k).
\label{eq:sum-VU}
\end{align}
Since $\bar{\mathbf{U}}^k = \mathbf{U}^k + \beta(\mathbf{U}^k-\mathbf{U}^{k-1})$ and $\bar{\mathbf{V}}^{k+1}=\mathbf{V}^{k+1} + \alpha(\mathbf{V}^{k+1}-\mathbf{V}^{k})$, we know that
\begin{equation}
\begin{aligned}
&f(\mathbf V^{k+1},\mathbf U^{k+1})
+\lambda_L\|\mathbf V^{k+1}\|_*+\lambda_S\|\mathbf W\mathbf U^{k+1}\|_1
\\
&\quad
+\frac{\rho_L}{2}\|\mathbf V^{k+1}-\mathbf L^k+\mathbf D_L^k\|_F^2+\frac{\rho_S}{2}\|\mathbf U^{k+1}-\mathbf S^k+\mathbf D_S^k\|_F^2
\\
\le\;&
f(\mathbf V^k,\mathbf U^k)
+\lambda_L\|\mathbf V^{k}\|_*+\lambda_S\|\mathbf W\mathbf U^{k}\|_1
\quad\\
&\quad
+\frac{\rho_L}{2}\|\mathbf V^{k}-\mathbf L^k+\mathbf D_L^k\|_F^2
+\frac{\rho_S}{2}\|\mathbf U^{k}-\mathbf S^k+\mathbf D_S^k\|_F^2+R_k
\label{eq:sum-VU}
\end{aligned}
\end{equation}
where $R_k=
f(\bar{\mathbf V}^{k+1},\mathbf U^{k})
-f(\mathbf V^{k+1},\bar{\mathbf U}^k)
+f(\mathbf V^{k},\bar{\mathbf U}^k)
-f(\mathbf V^k,\mathbf U^k)
+f(\mathbf V^{k+1},\mathbf U^{k+1})
-f(\bar{\mathbf V}^{k+1},\mathbf U^{k+1})
$. 
By Lemma \ref{descent property}, there exists a constant $L_f>0$ such that
\begin{equation}
\begin{aligned}
    R_k&=
f(\bar{\mathbf V}^{k+1},\mathbf U^{k})-f(\mathbf V^{k+1},\mathbf U^k)
+f(\mathbf V^{k+1},\mathbf U^k)-f(\mathbf V^{k+1},\bar{\mathbf U}^k)\\
&\quad+f(\mathbf V^{k},\bar{\mathbf U}^k)
-f(\mathbf V^k,\mathbf U^k)
+f(\mathbf V^{k+1},\mathbf U^{k+1})
-f(\bar{\mathbf V}^{k+1},\mathbf U^{k+1})\\
&\leq\left\langle \nabla_{\mathbf V} f(\mathbf V^{k+1},\mathbf U^k),\,
\bar{\mathbf V}^{k+1}-\mathbf V^{k+1}\right\rangle
+\frac{L_f}{2}\|\bar{\mathbf V}^{k+1}-\mathbf V^{k+1}\|_F^2\\
&\quad+\left\langle \nabla_{\mathbf U} f(\mathbf V^{k+1},\mathbf U^k),\,\mathbf U^{k}-\bar{\mathbf U}^{k}\right\rangle
+\frac{L_f}{2}\|\mathbf U^{k}-\bar{\mathbf U}^{k}\|_F^2
\\
&\quad+\left\langle \nabla_{\mathbf U} f(\mathbf V^{k},\mathbf U^k),\,
\bar{\mathbf U}^{k}-\mathbf U^{k}\right\rangle
+\frac{L_f}{2}\|\bar{\mathbf U}^{k}-\mathbf U^{k}\|_F^2 
\\
&\quad+\left\langle \nabla_{\mathbf V} f(\bar{\mathbf V}^{k+1},\mathbf U^{k+1}),\,
\mathbf V^{k+1}-\bar{\mathbf V}^{k+1}\right\rangle
+\frac{L_f}{2}\|\mathbf V^{k+1}-\bar{\mathbf V}^{k+1}\|_F^2.
\end{aligned}
\end{equation}
Since the extrapolated variables $\bar{\mathbf{U}}^k = \mathbf{U}^k + \beta(\mathbf{U}^k-\mathbf{U}^{k-1})$ and $\bar{\mathbf{V}}^{k+1}=\mathbf{V}^{k+1} + \alpha(\mathbf{V}^{k+1}-\mathbf{V}^{k})$, we have 
\begin{equation}
\begin{aligned}
R_k
=&\;
f(\bar{\mathbf V}^{k+1},\mathbf U^{k})-f(\mathbf V^{k+1},\mathbf U^k)
+f(\mathbf V^{k+1},\mathbf U^k)-f(\mathbf V^{k+1},\bar{\mathbf U}^k) \\
&\quad
+f(\mathbf V^{k},\bar{\mathbf U}^k)-f(\mathbf V^k,\mathbf U^k)
+f(\mathbf V^{k+1},\mathbf U^{k+1})-f(\bar{\mathbf V}^{k+1},\mathbf U^{k+1})\\
\leq &\;
\left\langle \nabla_{\mathbf V} f(\mathbf V^{k+1},\mathbf U^k),\,
 \alpha(\mathbf{V}^{k+1}-\mathbf{V}^{k})\right\rangle
+\frac{\alpha^2 L_f}{2}\|\mathbf{V}^{k+1}-\mathbf{V}^{k}\|_F^2\\
&\quad
+\left\langle \nabla_{\mathbf U} f(\mathbf V^{k+1},\bar{\mathbf U}^k),\,- \beta(\mathbf{U}^k-\mathbf{U}^{k-1})\right\rangle
+\frac{\beta^2 L_f}{2}\|\mathbf{U}^k-\mathbf{U}^{k-1}\|_F^2
\\
&\quad
+\left\langle \nabla_{\mathbf U} f(\mathbf V^{k},\mathbf U^k),\,
\beta(\mathbf{U}^k-\mathbf{U}^{k-1})\right\rangle
+\frac{\beta^2 L_f}{2}\|\mathbf{U}^k-\mathbf{U}^{k-1}\|_F^2
\\
&\quad
+\left\langle \nabla_{\mathbf V} f(\bar{\mathbf V}^{k+1},\mathbf U^{k+1}),\,
 - \alpha(\mathbf{V}^{k+1}-\mathbf{V}^{k})\right\rangle
+\frac{\alpha^2 L_f}{2}\| \mathbf{V}^{k+1}-\mathbf{V}^{k}\|_F^2.
\end{aligned}
\end{equation}
Hence, we obtain
\begin{equation}
\begin{aligned}
R_k
=&\;
f(\bar{\mathbf V}^{k+1},\mathbf U^{k})-f(\mathbf V^{k+1},\mathbf U^k)
+f(\mathbf V^{k+1},\mathbf U^k)-f(\mathbf V^{k+1},\bar{\mathbf U}^k)\\
&\quad
+f(\mathbf V^{k},\bar{\mathbf U}^k)-f(\mathbf V^k,\mathbf U^k)
+f(\mathbf V^{k+1},\mathbf U^{k+1})-f(\bar{\mathbf V}^{k+1},\mathbf U^{k+1})\\
\leq &\;
\alpha\left\langle \nabla_{\mathbf V} f(\mathbf V^{k+1},\mathbf U^k)-\nabla_{\mathbf V} f(\bar{\mathbf V}^{k+1},\mathbf U^{k+1}),\,
 \mathbf{V}^{k+1}-\mathbf{V}^{k}\right\rangle
\\
&\quad
+\beta\left\langle \nabla_{\mathbf U} f(\mathbf V^{k},\mathbf U^k)-\nabla_{\mathbf U} f(\mathbf V^{k+1},\bar{\mathbf U}^k),\, \mathbf{U}^k-\mathbf{U}^{k-1}\right\rangle
\\
&\quad
+{\alpha^2 L_f}\|\mathbf{V}^{k+1}-\mathbf{V}^{k}\|_F^2
+{\beta^2 L_f}\|\mathbf{U}^k-\mathbf{U}^{k-1}\|_F^2. 
\end{aligned}
\end{equation}
Using the Lipschitz continuity of $\nabla f$ on bounded sets, we have
\begin{equation}
\begin{aligned}
&\left\|
\nabla_{\mathbf V} f(\mathbf V^{k+1},\mathbf U^k)
-\nabla_{\mathbf V} f(\bar{\mathbf V}^{k+1},\mathbf U^{k+1})
\right\|
\\
&\le
L_f\left(
\|\mathbf V^{k+1}-\bar{\mathbf V}^{k+1}\|_F
+\|\mathbf U^k-\mathbf U^{k+1}\|_F
\right)
\\
&=
L_f\left(
\alpha\|\mathbf V^{k+1}-\mathbf V^k\|_F
+\|\mathbf U^{k+1}-\mathbf U^k\|_F
\right),
\end{aligned}
\end{equation}
and
\begin{equation}
\begin{aligned}
&\left\|
\nabla_{\mathbf U} f(\mathbf V^{k},\mathbf U^k)
-\nabla_{\mathbf U} f(\mathbf V^{k+1},\bar{\mathbf U}^k)
\right\|
\\
&\le
L_f\left(
\|\mathbf V^{k}-\mathbf V^{k+1}\|_F
+\|\mathbf U^k-\bar{\mathbf U}^k\|_F
\right)
\\
&=
L_f\left(
\|\mathbf V^{k+1}-\mathbf V^k\|_F
+\beta\|\mathbf U^{k}-\mathbf U^{k-1}\|_F
\right).
\end{aligned}
\end{equation}
From Cauchy-Schwarz inequality ($\langle a, b\rangle\le\Vert a\Vert\Vert b\Vert$), we have 
\begin{equation}
\begin{aligned}
&\alpha \Big\langle 
\nabla_{\mathbf V} f(\mathbf V^{k+1},\mathbf U^k)
-\nabla_{\mathbf V} f(\bar{\mathbf V}^{k+1},\mathbf U^{k+1}),
\,\mathbf V^{k+1}-\mathbf V^k
\Big\rangle 
\\
&\le
\alpha
\Big\|
\nabla_{\mathbf V} f(\mathbf V^{k+1},\mathbf U^k)
-\nabla_{\mathbf V} f(\bar{\mathbf V}^{k+1},\mathbf U^{k+1})
\Big\|
\cdot
\|\mathbf V^{k+1}-\mathbf V^k\|_F
\\
&\le
\alpha L_f
\Big(
\alpha\|\mathbf V^{k+1}-\mathbf V^k\|_F
+\|\mathbf U^{k+1}-\mathbf U^k\|_F
\Big)
\|\mathbf V^{k+1}-\mathbf V^k\|_F
\\
&=
\alpha^2L_f\|\mathbf V^{k+1}-\mathbf V^k\|_F^2
+\alpha L_f\|\mathbf V^{k+1}-\mathbf V^k\|_F\|\mathbf U^{k+1}-\mathbf U^k\|_F.
\end{aligned}
\end{equation}
and then
\begin{equation}
\begin{aligned}
R_k
\leq &\;
\alpha L_f\left(
\alpha\|\mathbf V^{k+1}-\mathbf V^k\|_F
+\|\mathbf U^{k+1}-\mathbf U^k\|_F
\right)\|\mathbf V^{k+1}-\mathbf V^k\|_F
\\
&\quad
+\beta L_f\left(
\|\mathbf V^{k+1}-\mathbf V^k\|_F
+\beta\|\mathbf U^{k}-\mathbf U^{k-1}\|_F
\right)\|\mathbf U^k-\mathbf U^{k-1}\|_F
\\
&\quad
+\alpha^2 L_f\|\mathbf V^{k+1}-\mathbf V^k\|_F^2
+\beta^2 L_f\|\mathbf U^k-\mathbf U^{k-1}\|_F^2\\
=&
2\alpha^2 L_f\|\mathbf V^{k+1}-\mathbf V^k\|_F^2
+\alpha L_f\|\mathbf V^{k+1}-\mathbf V^k\|_F\|\mathbf U^{k+1}-\mathbf U^k\|_F
\\
&\quad
+\beta L_f\|\mathbf V^{k+1}-\mathbf V^k\|_F\|\mathbf U^{k}-\mathbf U^{k-1}\|_F
+2\beta^2 L_f\|\mathbf U^{k}-\mathbf U^{k-1}\|_F^2.
\end{aligned}
\end{equation}
Applying Young's inequality ($ab<\frac{\eta}{2}a^2+\frac{1}{2\eta}b^2$), for any $\eta_1,\eta_2>0$, we obtain
\begin{equation}
\begin{aligned}
R_k
\leq &\;
\left(2\alpha^2L_f+\frac{\alpha L_f}{2}\eta_1+\frac{\beta L_f}{2}\eta_2\right)
\|\mathbf V^{k+1}-\mathbf V^k\|_F^2
\\
&\quad
+\frac{\alpha L_f}{2\eta_1}\|\mathbf U^{k+1}-\mathbf U^k\|_F^2
+\left(2\beta^2L_f+\frac{\beta L_f}{2\eta_2}\right)
\|\mathbf U^{k}-\mathbf U^{k-1}\|_F^2.
\end{aligned}
\end{equation}
Substituting the above bound on $R_k$ into \eqref{eq:sum-VU}, we obtain
\begin{equation}
\begin{aligned}
&f(\mathbf V^{k+1},\mathbf U^{k+1})
+\lambda_L\|\mathbf V^{k+1}\|_*+\lambda_S\|\mathbf W\mathbf U^{k+1}\|_1
\\
&\quad
+\frac{\rho_L}{2}\|\mathbf V^{k+1}-\mathbf L^k+\mathbf D_L^k\|_F^2+\frac{\rho_S}{2}\|\mathbf U^{k+1}-\mathbf S^k+\mathbf D_S^k\|_F^2
\\
\le\;&
f(\mathbf V^k,\mathbf U^k)
+\lambda_L\|\mathbf V^{k}\|_*+\lambda_S\|\mathbf W\mathbf U^{k}\|_1
\quad\\
&\quad
+\frac{\rho_L}{2}\|\mathbf V^{k}-\mathbf L^k+\mathbf D_L^k\|_F^2
+\frac{\rho_S}{2}\|\mathbf U^{k}-\mathbf S^k+\mathbf D_S^k\|_F^2
\\
&\quad
+c_1\|\mathbf V^{k+1}-\mathbf V^k\|_F^2
+c_2\|\mathbf U^{k+1}-\mathbf U^k\|_F^2
+c_3\|\mathbf U^{k}-\mathbf U^{k-1}\|_F^2,
\label{eq:pre-deep}
\end{aligned}
\end{equation}
where $c_1:= 2\alpha^2L_f+\frac{\alpha L_f}{2}\eta_1+\frac{\beta L_f}{2}\eta_2$, $c_2:=\frac{\alpha L_f}{2\eta_1}$, and $c_3:=2\beta^2L_f+\frac{\beta L_f}{2\eta_2}$. 

On the other hand, denote $\textbf{L}^{k}:=\mathbf L(\bm{\theta}_L^{k},\textbf{z}_L)$ and $\textbf{S}^k:=\mathbf S(\bm{\theta}_S^{k},\textbf{z}_S)$, by Assumption~\ref{ass3:inexact}, the updates of the $\textbf{L}$ and $\textbf{S}$ subproblems satisfy
\begin{equation}
\begin{aligned}
&\frac{\rho_L}{2}\|\mathbf V^{k+1}-\mathbf L^{k+1}+\mathbf D_L^k\|_F^2
\le
\frac{\rho_L}{2}\|\mathbf V^{k+1}-\mathbf L^{k}+\mathbf D_L^k\|_F^2
+\varepsilon_k^L,
\\
&\frac{\rho_S}{2}\|\mathbf U^{k+1}-\mathbf S^{k+1}+\mathbf D_S^k\|_F^2
\le\
\frac{\rho_S}{2}\|\mathbf U^{k+1}-\mathbf S^{k}+\mathbf D_S^k\|_F^2
+\varepsilon_k^S.
\label{eq:deep-S-descent}
\end{aligned}
\end{equation}
Combining \eqref{eq:pre-deep} with \eqref{eq:deep-S-descent}, we obtain
\begin{equation}
\begin{aligned}
&f(\mathbf V^{k+1},\mathbf U^{k+1})
+\lambda_L\|\mathbf V^{k+1}\|_*
+\lambda_S\|\mathbf W\mathbf U^{k+1}\|_1
\\
&\quad
+\frac{\rho_L}{2}\|\mathbf V^{k+1}-\mathbf L^{k+1}+\mathbf D_L^k\|_F^2
+\frac{\rho_S}{2}\|\mathbf U^{k+1}-\mathbf S^{k+1}+\mathbf D_S^k\|_F^2
\\
\le\;&
f(\mathbf V^k,\mathbf U^k)
+\lambda_L\|\mathbf V^{k}\|_*
+\lambda_S\|\mathbf W\mathbf U^{k}\|_1
\\
&\quad
+\frac{\rho_L}{2}\|\mathbf V^{k}-\mathbf L^k+\mathbf D_L^k\|_F^2
+\frac{\rho_S}{2}\|\mathbf U^{k}-\mathbf S^k+\mathbf D_S^k\|_F^2
\\
&\quad
+c_1\|\mathbf V^{k+1}-\mathbf V^k\|_F^2
+c_2\|\mathbf U^{k+1}-\mathbf U^k\|_F^2
+c_3\|\mathbf U^{k}-\mathbf U^{k-1}\|_F^2
+\varepsilon_k^L+\varepsilon_k^S.
\end{aligned}
\label{eq:before-dual}
\end{equation}
Since 
\begin{equation}
\mathbf D_L^{k+1}=\mathbf D_L^k+\mathbf V^{k+1}-\mathbf L^{k+1},
\qquad
\mathbf D_S^{k+1}=\mathbf D_S^k+\mathbf U^{k+1}-\mathbf S^{k+1},
\label{eq:dual-update}
\end{equation}
we have
\begin{equation}
\frac{\rho_L}{2}\|\mathbf V^{k+1}-\mathbf L^{k+1}+\mathbf D_L^k\|_F^2
-\frac{\rho_L}{2}\|\mathbf D_L^k\|_F^2
=
\frac{\rho_L}{2}\|\mathbf D_L^{k+1}\|_F^2
-\frac{\rho_L}{2}\|\mathbf D_L^k\|_F^2,
\label{eq:Dl-identity}
\end{equation}
and similarly,
\begin{equation}
\frac{\rho_S}{2}\|\mathbf U^{k+1}-\mathbf S^{k+1}+\mathbf D_S^k\|_F^2
-\frac{\rho_S}{2}\|\mathbf D_S^k\|_F^2
=
\frac{\rho_S}{2}\|\mathbf D_S^{k+1}\|_F^2
-\frac{\rho_S}{2}\|\mathbf D_S^k\|_F^2.
\label{eq:Ds-identity}
\end{equation}
Recalling the definition of the augmented Lagrangian \eqref{eq:aug_lag}, from \eqref{eq:before-dual} and \eqref{eq:Ds-identity} we obtain
\begin{equation}
\begin{aligned}
&
\mathcal L^{k+1}
\le\;\mathcal L^{k}
+c_1\|\mathbf V^{k+1}-\mathbf V^k\|_F^2
+c_2\|\mathbf U^{k+1}-\mathbf U^k\|_F^2
+c_3\|\mathbf U^{k}-\mathbf U^{k-1}\|_F^2
+\varepsilon_k^L+\varepsilon_k^S.
\end{aligned}
\label{eq:AL-descent}
\end{equation}
Since the $\mathbf V$- and $\mathbf U$-subproblems are strongly convex with moduli $\tau_V$ and $\tau_U$, respectively, 
we know that 
\begin{equation}\label{eq:AL-descent2}
\begin{aligned}
&
\mathcal L^{k+1}
\le\;
\mathcal L^{k}
-\Big(\frac{\tau_V}{2}-c_1\Big)\|\mathbf V^{k+1}-\mathbf V^k\|_F^2
-\Big(\frac{\tau_U}{2}-c_2\Big)\|\mathbf U^{k+1}-\mathbf U^k\|_F^2
\\
&\quad\qquad
+c_3\|\mathbf U^{k}-\mathbf U^{k-1}\|_F^2
+\varepsilon_k^L+\varepsilon_k^S.
\end{aligned}
\end{equation}
Adding $\frac{\delta_V}{2}\|\mathbf V^{k+1}-\mathbf V^k\|_F^2+\frac{\delta_U}{2}\|\mathbf U^{k+1}-\mathbf U^k\|_F^2$ to both sides of \eqref{eq:AL-descent2}, and subtracting $\frac{\delta_V}{2}\|\mathbf V^{k}-\mathbf V^{k-1}\|_F^2 +\frac{\delta_U}{2}\|\mathbf U^{k}-\mathbf U^{k-1}\|_F^2$,
we obtain
\begin{equation}
\begin{aligned}
\Psi^{k+1}-\Psi^k
\le\;&
-\left(\frac{\tau_V}{2}-\frac{\delta_V}{2}-c_1\right)\|\mathbf V^{k+1}-\mathbf V^k\|_F^2
-\left(\frac{\tau_U}{2}-\frac{\delta_U}{2}-c_2\right)\|\mathbf U^{k+1}-\mathbf U^k\|_F^2\\
&
-\frac{\delta_V}{2}\|\mathbf V^{k}-\mathbf V^{k-1}\|_F^2
-\left(\frac{\delta_U}{2}-c_3\right)\|\mathbf U^{k}-\mathbf U^{k-1}\|_F^2
+\varepsilon_k^L+\varepsilon_k^S.
\end{aligned}
\label{eq:psi-step}
\end{equation}
Let $\gamma_k:=\varepsilon_k^L+\varepsilon_k^S$, since $\sum_{k=0}^\infty \varepsilon_k^L<\infty$ and $\sum_{k=0}^\infty \varepsilon_k^S<\infty$, we have $\sum_{k=0}^\infty \gamma_k<\infty$.
Choose $\delta_V$ and $\delta_U$ such that
$\delta_V<\tau_V-2c_1$, $2c_3<\delta_U<\tau_U-2c_2$.
Then there exist constants $\xi_V,\xi_U>0$ such that
\begin{equation}
\Psi^{k+1}-\Psi^k
\le
-\xi_V\|\mathbf V^{k+1}-\mathbf V^k\|_F^2
-\xi_U\|\mathbf U^{k+1}-\mathbf U^k\|_F^2
+\gamma_k.
\label{eq:final-descent}
\end{equation}
This proves the sufficient descent property.
\end{proof}

\begin{remark}
The sufficient descent condition requires
\[
c_1 < \frac{\tau_V-\delta_V}{2},\qquad
c_2 < \frac{\tau_U-\delta_U}{2},\qquad
c_3 < \frac{\delta_U}{2},
\]
where
\[
c_1=2\alpha^2L_f+\frac{\alpha L_f}{2}\eta_1+\frac{\beta L_f}{2}\eta_2,\quad
c_2=\frac{\alpha L_f}{2\eta_1},\quad
c_3=2\beta^2L_f+\frac{\beta L_f}{2\eta_2},
\]
strongly convex moduli $\tau_V$ and $\tau_U$ are chosen sufficiently large, and $\eta_1, \eta_2>0$.
Hence, there exists a nonempty neighborhood of $(\alpha,\beta)$ such that the sufficient descent property holds for all extrapolation parameters in this neighborhood. 
In particular, a sufficient choice is
\[
0\le \alpha<\frac{\eta_1(\tau_U-\delta_U)}{L_f},
\qquad
0\le \beta<\beta_{\max},
\]
where $\beta_{\max}$ is the positive root of $2L_f\beta^2+\frac{L_f}{2\eta_2}\beta-\frac{\delta_U}{2}=0$.
\end{remark}

As a direct consequence of Lemma~\ref{lem:descent}, we have the following lemma. 
\begin{lemma}\label{lem:asymptotic}
Suppose Assumptions~\ref{ass2:Abounded}, \ref{ass3:inexact}, and \ref{ass:param} hold. Then
\begin{equation}\label{eq:square-summable}
\sum_{k=0}^{\infty}\|\mathbf V^{k+1}-\mathbf V^k\|_F^2<\infty,
\qquad
\sum_{k=0}^{\infty}\|\mathbf U^{k+1}-\mathbf U^k\|_F^2<\infty.
\end{equation}
Consequently,
\begin{equation}\label{eq:vanishing-VU}
\|\mathbf V^{k+1}-\mathbf V^k\|_F\to0,
\qquad
\|\mathbf U^{k+1}-\mathbf U^k\|_F\to0.
\end{equation}
\end{lemma}

\begin{proof}
From Lemma~\ref{lem:descent}, we have
\[
\Psi^{k+1}-\Psi^k
\le
-\xi_V\|\mathbf V^{k+1}-\mathbf V^k\|_F^2
-\xi_U\|\mathbf U^{k+1}-\mathbf U^k\|_F^2
+\gamma_k.
\]
Summing both sides from $k=0$ to $N$, we obtain
\begin{align*}
\Psi^{N+1}-\Psi^0
\le
& -\xi_V \sum_{k=0}^{N}\|\mathbf V^{k+1}-\mathbf V^k\|_F^2 \\
& -\xi_U \sum_{k=0}^{N}\|\mathbf U^{k+1}-\mathbf U^k\|_F^2
+ \sum_{k=0}^{N}\gamma_k.
\end{align*}
Since $\Psi^k$ is bounded from below and $\sum_{k=0}^\infty \gamma_k<\infty$, letting $N\to\infty$ yields \eqref{eq:square-summable}. The limits in \eqref{eq:vanishing-VU} follow immediately.
\end{proof}

We next show that the limiting subgradient of the Lyapunov function can be controlled by successive differences of the iterates.

\begin{lemma}\label{lem:subgradient}
Suppose Assumptions~\ref{ass2:Abounded}, \ref{ass3:inexact}, and \ref{ass:param} hold. Denote $\textbf{L}^{k}:=\mathbf L(\bm{\theta}_L^{k},\textbf{z}_L)$ and $\textbf{S}^k:=\mathbf S(\bm{\theta}_S^{k},\textbf{z}_S)$, then there exist a constant $C>0$ and a summable sequence $\{\gamma_k\}$ such that
\begin{equation}
\begin{aligned}
\mathrm{dist}(0,\partial \Psi^{k+1})
\le\;&
C\Big(
\|\mathbf V^{k+1}-\mathbf V^k\|_F
+\|\mathbf U^{k+1}-\mathbf U^k\|_F
+\|\mathbf U^k-\mathbf U^{k-1}\|_F
\\&
+\|\mathbf L^{k+1}-\mathbf L^k\|_F
+\|\mathbf S^{k+1}-\mathbf S^k\|_F
+\|\mathbf D_L^{k+1}-\mathbf D_L^k\|_F
+\|\mathbf D_S^{k+1}-\mathbf D_S^k\|_F
\Big)
+\gamma_k.
\end{aligned}
\label{eq:subgrad}
\end{equation}
\end{lemma}
\begin{proof}
We estimate each block of the limiting subgradient of $\Psi^{k+1}$. 
Recall that
\begin{equation}
\Psi^{k+1}=
\mathcal L^{k+1}
+\frac{\delta_V}{2}\|\mathbf V^{k+1}-\mathbf V^k\|_F^2
+\frac{\delta_U}{2}\|\mathbf U^{k+1}-\mathbf U^k\|_F^2.
\end{equation}
First, by the optimality condition of the $\mathbf V$-subproblem, there exists
$\mathbf G_V^{k+1}\in \partial \|\mathbf V^{k+1}\|_*$
such that
\begin{equation}
\nabla_{\mathbf V} f(\mathbf V^{k+1},\bar{\mathbf U}^k)
+\lambda_L \mathbf G_V^{k+1}
+\rho_L\bigl(\mathbf V^{k+1}-\mathbf L^k+\mathbf D_L^k\bigr)=0.
\label{eq:opt-V}
\end{equation}
On the other hand, the $\mathbf V$-component of the limiting subgradient of $\Psi^{k+1}$ is
\begin{equation}
\partial_{\mathbf V}\Psi^{k+1}
=
\nabla_{\mathbf V} f(\mathbf V^{k+1},\mathbf U^{k+1})
+\lambda_L \partial \|\mathbf V^{k+1}\|_*
+\rho_L\bigl(\mathbf V^{k+1}-\mathbf L^{k+1}+\mathbf D_L^{k+1}\bigr)
+\delta_V(\mathbf V^{k+1}-\mathbf V^k).
\label{eq:subgrad-V}
\end{equation}
Using \eqref{eq:opt-V} and \eqref{eq:subgrad-V}, define
\begin{equation}
\begin{aligned}
\mathbf q_V^{k+1}
:=
&\;\nabla_{\mathbf V} f(\mathbf V^{k+1},\mathbf U^{k+1})
-\nabla_{\mathbf V} f(\mathbf V^{k+1},\bar{\mathbf U}^k)
\\
&+\rho_L\Big[
(\mathbf V^{k+1}-\mathbf L^{k+1}+\mathbf D_L^{k+1})
-(\mathbf V^{k+1}-\mathbf L^k+\mathbf D_L^k)
\Big]
+\delta_V(\mathbf V^{k+1}-\mathbf V^k).
\end{aligned}
\label{eq:qV-def}
\end{equation}
Then, we have
$\mathbf q_V^{k+1}\in \partial_{\mathbf V}\Psi^{k+1}$.
We now bound its norm and get
\begin{equation}
\begin{aligned}
\|\mathbf q_V^{k+1}\|_F
\le\;&L_f\|\mathbf U^{k+1}-\bar{\mathbf U}^k\|_F
+\rho_L\Vert (\mathbf L^k-\mathbf L^{k+1})+(\mathbf D_L^{k+1}-\mathbf D_L^k)\Vert_F+\delta_V\|\mathbf V^{k+1}-\mathbf V^k\|_F
\\
\le\;&
L_f\|\mathbf U^{k+1}-\mathbf U^k\|_F
+\beta L_f\|\mathbf U^k-\mathbf U^{k-1}\|_F
\\
&+\rho_L\|\mathbf L^{k+1}-\mathbf L^k\|_F
+\rho_L\|\mathbf D_L^{k+1}-\mathbf D_L^k\|_F
+\delta_V\|\mathbf V^{k+1}-\mathbf V^k\|_F.
\end{aligned}
\label{eq:qV-bound}
\end{equation}

Second, by the optimality condition of the $\mathbf U$-subproblem, there exists 
$\mathbf G_U^{k+1}\in \partial \|\mathbf W\mathbf U^{k+1}\|_1$
such that
\begin{equation}
\nabla_{\mathbf U} f(\bar{\mathbf V}^{k+1},\mathbf U^{k+1})
+\lambda_S \mathbf G_U^{k+1}
+\rho_S\bigl(\mathbf U^{k+1}-\mathbf S^k+\mathbf D_S^k\bigr)=0.
\label{eq:opt-U}
\end{equation}
The $\mathbf U$-component of $\partial \Psi^{k+1}$ is
\begin{equation}
\partial_{\mathbf U}\Psi^{k+1}
=
\nabla_{\mathbf U} f(\mathbf V^{k+1},\mathbf U^{k+1})
+\lambda_S \partial \|\mathbf W\mathbf U^{k+1}\|_1
+\rho_S\bigl(\mathbf U^{k+1}-\mathbf S^{k+1}+\mathbf D_S^{k+1}\bigr)
+\delta_U(\mathbf U^{k+1}-\mathbf U^k).
\label{eq:subgrad-U}
\end{equation}
Using \eqref{eq:opt-U} and \eqref{eq:subgrad-U}, define
\begin{equation}
\begin{aligned}
\mathbf q_U^{k+1}
:=
&\;\nabla_{\mathbf U} f(\mathbf V^{k+1},\mathbf U^{k+1})
-\nabla_{\mathbf U} f(\bar{\mathbf V}^{k+1},\mathbf U^{k+1})
\\
&+\rho_S\Big[
(\mathbf U^{k+1}-\mathbf S^{k+1}+\mathbf D_S^{k+1})
-(\mathbf U^{k+1}-\mathbf S^k+\mathbf D_S^k)
\Big]
+\delta_U(\mathbf U^{k+1}-\mathbf U^k).
\end{aligned}
\label{eq:qU-def}
\end{equation}
Then we know 
$\mathbf q_U^{k+1}\in \partial_{\mathbf U}\Psi^{k+1}$. 
Therefore,
\begin{equation}
\begin{aligned}
\|\mathbf q_U^{k+1}\|_F
\leq\;& L_f\|\mathbf V^{k+1}-\bar{\mathbf V}^{k+1}\|_F+\rho_S\Vert(\mathbf S^k-\mathbf S^{k+1})+(\mathbf D_S^{k+1}-\mathbf D_S^k)\Vert_F+\delta_U\|\mathbf U^{k+1}-\mathbf U^k\|_F
\\
\le\;&
\alpha L_f\|\mathbf V^{k+1}-\mathbf V^k\|_F
+\rho_S\|\mathbf S^{k+1}-\mathbf S^k\|_F
+\rho_S\|\mathbf D_S^{k+1}-\mathbf D_S^k\|_F+\delta_U\|\mathbf U^{k+1}-\mathbf U^k\|_F.
\end{aligned}
\label{eq:qU-bound}
\end{equation}
Third, from the augmented Lagrangian,
\[
\partial_{\mathbf D_L}\Psi^{k+1}
=
\rho_L(\mathbf V^{k+1}-\mathbf L^{k+1}),
\qquad
\partial_{\mathbf D_S}\Psi^{k+1}
=
\rho_S(\mathbf U^{k+1}-\mathbf S^{k+1}).
\]
Since
\[
\mathbf D_L^{k+1}-\mathbf D_L^k=\mathbf V^{k+1}-\mathbf L^{k+1},
\qquad
\mathbf D_S^{k+1}-\mathbf D_S^k=\mathbf U^{k+1}-\mathbf S^{k+1},
\]
we obtain
\begin{equation}
\|\partial_{\mathbf D_L}\Psi^{k+1}\|_F
=
\rho_L\|\mathbf D_L^{k+1}-\mathbf D_L^k\|_F,
\qquad
\|\partial_{\mathbf D_S}\Psi^{k+1}\|_F
=
\rho_S\|\mathbf D_S^{k+1}-\mathbf D_S^k\|_F.
\label{eq:dual-block-bound}
\end{equation}
Finally, from Assumption~\ref{ass3:inexact}, we know that
\begin{equation}
\mathrm{dist}\bigl(0,\partial \Phi_L^k(\bm{\theta}_L^{k+1})\bigr)\le \varepsilon_k^L,
\qquad
\mathrm{dist}\bigl(0,\partial \Phi_S^k(\bm{\theta}_S^{k+1})\bigr)\le \varepsilon_k^S.
\end{equation}
Hence there exist
\begin{equation}
\mathbf q_{\bm{\theta}_L}^{k+1}\in \partial_{\bm{\theta}_L}\Psi^{k+1},
\qquad
\mathbf q_{\bm{\theta}_S}^{k+1}\in \partial_{\bm{\theta}_S}\Psi^{k+1},
\end{equation}
such that
\begin{equation}\label{eq:net-bound}
\|\mathbf q_{\bm{\theta}_L}^{k+1}\|\le \varepsilon_k^L,
\qquad
\|\mathbf q_{\bm{\theta}_S}^{k+1}\|\le \varepsilon_k^S.
\end{equation}
Collecting \eqref{eq:qV-bound}, \eqref{eq:qU-bound}, \eqref{eq:dual-block-bound}, and \eqref{eq:net-bound} and absorbing fixed coefficients into a generic constant $C>0$, we conclude that
\eqref{eq:subgrad} hold
with $\gamma_k:=\varepsilon_k^L+\varepsilon_k^S$.
\end{proof}

\begin{theorem}\label{thm:basic}
Suppose Assumptions~\ref{ass2:Abounded}, \ref{ass3:inexact}, and \ref{ass:param} hold. 
Let the sequence

\[
\left\{(\mathbf V^k,\mathbf U^k,\bm{\theta}_L^k,\bm{\theta}_S^k,\mathbf D_L^k,\mathbf D_S^k)\right\}
\]

be generated by Algorithm~\ref{alg:eadmm}. Then:
\begin{itemize}
\item[(i)] The set of cluster points of the sequence is nonempty and compact.

\item[(ii)] The sequence $\{\Psi^k\}$ converges. 

\end{itemize}
\end{theorem}
\begin{proof}
(i) By Assumption~\ref{ass2:Abounded}, the sequence
$\{(\mathbf V^k,\mathbf U^k,\bm{\theta}_L^k,\bm{\theta}_S^k,\mathbf D_L^k,\mathbf D_S^k)\}_{k\ge0}$
is bounded. Hence, by the Bolzano--Weierstrass theorem, it admits at least one cluster point. Therefore, the cluster point set is nonempty. Since the cluster point set of a bounded sequence is closed and bounded, it is compact.

(ii) From Lemma~\ref{lem:descent}, we have
\begin{equation}
\Psi^{k+1}-\Psi^k
\le
-\xi_V\|\mathbf V^{k+1}-\mathbf V^k\|_F^2
-\xi_U\|\mathbf U^{k+1}-\mathbf U^k\|_F^2
+\gamma_k,
\end{equation}
where $\sum_{k=0}^{\infty}\gamma_k<\infty$.
Define
\begin{equation}
\widehat\Psi^k
:=
\Psi^k+\sum_{j=k}^{\infty}\gamma_j.
\end{equation}
Then
\begin{equation}
\widehat\Psi^{k+1}
=
\Psi^{k+1}+\sum_{j=k+1}^{\infty}\gamma_j
\le
\Psi^k+\sum_{j=k}^{\infty}\gamma_j
=
\widehat\Psi^k.
\end{equation}
Hence, $\{\widehat\Psi^k\}$ is nonincreasing. Since $\Psi^k$ is bounded from below and $\sum_{j=k}^{\infty}\gamma_j\ge0$, the sequence $\{\widehat\Psi^k\}$ is bounded from below and therefore converges. Because
\begin{equation}
\sum_{j=k}^{\infty}\gamma_j\to0,
\end{equation}
we conclude that $\{\Psi^k\}$ also converges.

\end{proof}

\begin{remark}\label{ass:auxiliary}
The estimate in Lemma~\ref{lem:subgradient} involves auxiliary increments
\[
\|\mathbf L^{k+1}-\mathbf L^k\|_F,\quad
\|\mathbf S^{k+1}-\mathbf S^k\|_F,\quad
\|\mathbf D_L^{k+1}-\mathbf D_L^k\|_F,\quad
\|\mathbf D_S^{k+1}-\mathbf D_S^k\|_F.
\]
Controlling or vanishing of these quantities is sufficient to ensure that the limiting subgradient tends to zero.
\end{remark}

\begin{corollary}\label{cor:critical}
If, in addition, the auxiliary increments in Remark~\ref{ass:auxiliary} vanish asymptotically, then every cluster point of the generated sequence is a critical point of $\Psi$.
\end{corollary}
\begin{proof}
From Lemma~\ref{lem:subgradient} and the auxiliary increments vanish asymptotically by assumption, and $\gamma_k\to 0$, we have
$
\mathrm{dist}(0,\partial \Psi^{k}) \to 0.
$
Let $\mathbf B^k := (\mathbf V^k,\mathbf U^k,\bm{\theta}_L^k,\bm{\theta}_S^k,\mathbf D_L^k,\mathbf D_S^k)$ denote the full iterate.
Let $\mathbf B^{k_j}\to \mathbf B^*$ be any convergent subsequence. Since $\Psi^k$ converges, we also have $\Psi(\mathbf B^{k_j}) \to \Psi(\mathbf B^*)$. 
By the closedness of the limiting subdifferential, we obtain $
0 \in \partial \Psi(\mathbf B^*)$.
Thus $\mathbf B^*$ is a critical point of $\Psi$.
\end{proof}
\begin{remark}
If the Lyapunov function $\Psi$ satisfies the Kurdyka--\L ojasiewicz (KL) property, and additional regularity conditions ensure sufficient control of auxiliary increments, then the standard KL framework can be applied to establish convergence of the sequence.
\end{remark}
\section{Limitations}\label{appen: limitations}

One limitation of the proposed method is its computational cost. Similar to other DIP-based approaches, the method requires scan-specific test-time optimization, resulting in longer reconstruction times compared with the inference time for supervised learning methods that require only a single or a few forward passes. In addition, compared with classical optimization-based approaches such as compressed sensing, the proposed L+S-DIP framework introduces additional hyperparameters that are related to network architecture, and network parameters optimization, which may require careful tuning for stable performance.
\section{Impact of the choice of the regularization parameters}\label{appen: impact of choice of reg params}
We study the sensitivity of the proposed method to the choice of the two regularization parameters $\lambda_L$ and $\lambda_S$, which control the strengths of the low-rank and sparse priors, respectively. Table~\ref{tab:lambda_grid} reports the PSNR results obtained with different combinations of $\lambda_L$ and $\lambda_S$.
The evaluation is conducted on the OCMR dataset with $4\times$ acceleration, using two slices.
From the results, we observe that our method performs well under most parameter settings. However, the setting $\lambda_L = 0.2, \lambda_S = 2\times10^{-4}$ returns the highest PSNR and therefore, we use it for all experiments..

\begin{table}[htp!]
\centering
\caption{
PSNR sensitivity to $\lambda_L$ and $\lambda_S$.
The best result is highlighted in bold.
}
\label{tab:lambda_grid}
\setlength{\tabcolsep}{5.5pt}
\renewcommand{\arraystretch}{1.12}
\resizebox{0.48\textwidth}{!}{
\begin{tabular}{c|cccc}
\toprule
\diagbox[width=6.2em,height=2.em]{$\lambda_S$}{$\lambda_L$}
& $0.05$ & $0.10$ & $0.20$ & $0.50$ \\
\midrule
$1\times 10^{-4}$ & 32.67 & 32.96 & 32.49 & 32.96 \\
$2\times 10^{-4}$ & 32.39 & 32.53 &  \textbf{33.83} & 33.69 \\
$5\times 10^{-4}$ & 32.54 & 32.51 & 32.19 & 32.17 \\
$1\times 10^{-3}$ & 32.55 & 32.15 & 32.57 & 31.55 \\
\bottomrule
\end{tabular}
}
\end{table}

\section{Uncertainty maps}\label{appen: uncertainty maps}


We further visualize reconstruction variability by running our method with different random network initializations.
Specifically, we use four random seeds and compute the mean and standard deviation of the reconstructed images across these runs.
Fig.~\ref{fig:uncertainty_maps} shows the ground truth, the $4\times$ undersampled/aliased reconstruction input, the mean reconstruction over four random seeds, and the corresponding standard deviation map.
The standard deviation map shows relatively small variation across random initializations, suggesting that the proposed method produces consistent reconstructions for this example.
However, some structured residual artifacts may still be visible in the reconstruction.

\begin{figure}[t]
    \centering
    \includegraphics[width=0.65\linewidth]{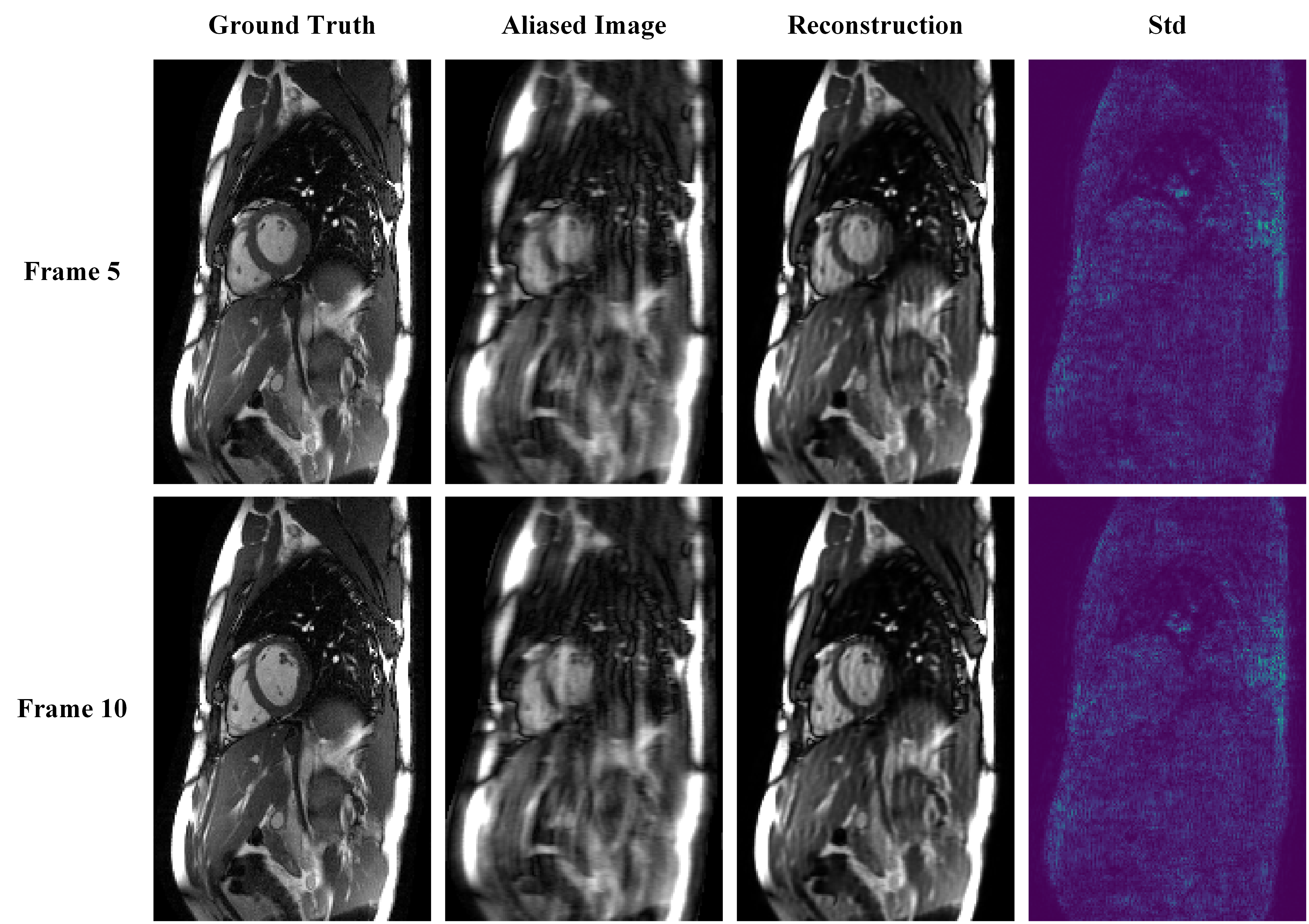}
\caption{
Effect of random network initialization on the proposed L+S-DIP reconstruction.
From left to right, the columns show the ground truth, the $4\times$ undersampled/aliased image, the mean reconstruction over four random seeds, and the corresponding standard deviation map.
The relatively low standard deviation indicates stable reconstructions across different initializations for this example.
}
    \label{fig:uncertainty_maps}
\end{figure}


\section{{Visualization of the Learned Low-Rank and Sparse Components}}\label{appen: low rank and sparse visuals}


To better understand how the proposed decomposition separates the learned low-rank and sparse components, we visualize both of these components separately in Figure~\ref{fig:low_sparse_visual}. The figure shows two temporal frames from the same cine slice, with columns corresponding to the ground truth, the 4$\times$ undersampled image, the L+S DIP reconstructed image, and the learned low-rank and sparse component. The reconstructed images closely recover the main cardiac anatomy compared with the aliased inputs, indicating that the proposed decomposition is effective for removing undersampling artifacts.
The learned low-rank component contains the dominant anatomical structures that are largely shared across time, such as the surrounding background anatomy, while the sparse component has lower overall intensity and emphasizes frame-dependent variations and localized dynamic details. Thus, the low rank part $\mathbf{L}$ captures the temporally correlated background, while $\mathbf{S}$ accounts for dynamic innovations that are not well represented by the low-rank component alone.

\begin{figure*}[t]
    \centering
    \includegraphics[width=0.8\textwidth]{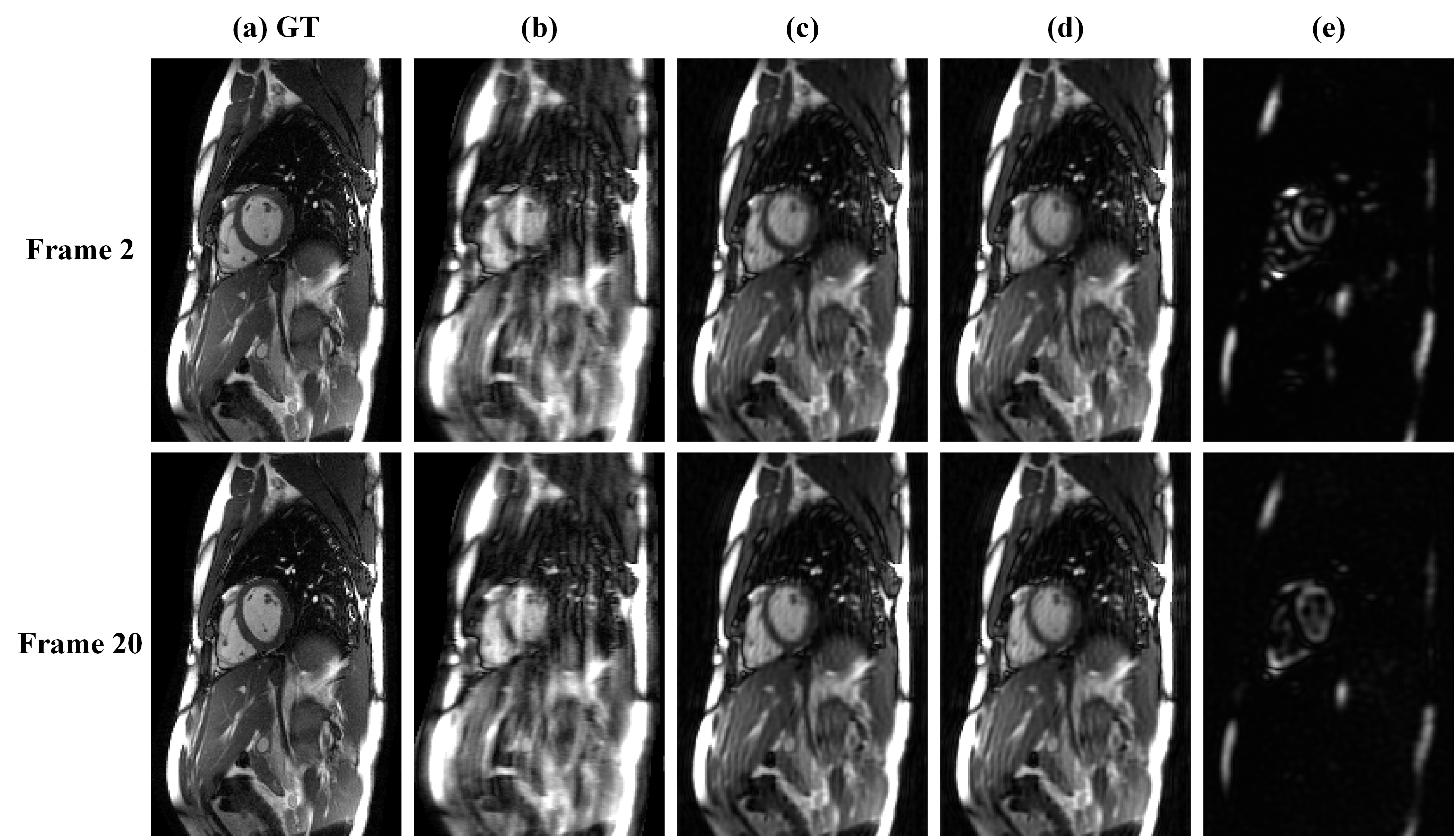}
    \caption{
    Visualization of the learned low-rank and sparse components for the OCMR dataset.
    Columns show (a) ground truth (GT), (b) $4 \times$ undersampled aliased image, (c) reconstructed image for L+S DIP, (d) learned low-rank component, and (e) learned sparse component.
    The two rows show two different frames from the same slice.
    }
    \label{fig:low_sparse_visual}
\end{figure*}